\documentclass{birkjour}
\usepackage[T1]{fontenc}
\usepackage[latin9]{inputenc}
\usepackage{geometry}
\usepackage{amsthm}
\usepackage{amsmath}
\usepackage{amssymb}
\usepackage{color}
\usepackage{hyperref}
\usepackage{graphicx}
\usepackage{epstopdf}
\usepackage{pdfpages}
\usepackage{mathtools}
\DeclarePairedDelimiter\ceil{\lceil}{\rceil}
\DeclarePairedDelimiter\floor{\lfloor}{\rfloor}

\def\mc {\mathcal}
\def\mk {\mathfrak}

\def\mr {\mathrm}
\def\mb {\mathbb}

\def\RR {{\mathbb R}}

\makeatletter

 \newtheorem{thm}{Theorem}[section]
 \newtheorem{coll}[thm]{Corollary}
 \newtheorem{lem}[thm]{Lemma}
 
  \newtheorem{fact}[thm]{Fact}
 \theoremstyle{definition}
 \newtheorem{definition}[thm]{Definition}
 \theoremstyle{remark}

 \numberwithin{equation}{section}
 
\begin{document}

\title{Universality of single qudit gates}

\author{Adam Sawicki}

\address{%
Center for Theoretical Physics PAS\\ 
Al. Lotnik\'ow 32/46\\
02-668, Warsaw\\
Poland}
\email{a.sawicki@cft.edu.pl}

\author{Katarzyna Karnas}
\address{%
Center for Theoretical Physics PAS\\ 
Al. Lotnik\'ow 32/46\\
02-668, Warsaw\\
Poland}
\email{karnas@cft.edu.pl}

\subjclass{Primary 81P45; Secondary 22E46}

\keywords{universal gates, universal Hamiltonians, Lie groups, Lie algebras, representation theory}

\begin{abstract}
We consider the problem of deciding if a set of quantum one-qudit gates $\mathcal{S}=\{g_1,\ldots,g_n\}\subset G$ is universal, i.e. if $<\mathcal{S}>$ is dense in $G$, where $G$ is either the special unitary or the special orthogonal group. To every gate $g$ in $\mathcal{S}$ we assign the orthogonal matrix $\mr{Ad}_g$ that is image of $g$ under the adjoint representation $\mathrm{Ad}:G\rightarrow SO(\mathfrak{g})$ and $\mathfrak{g}$ is the Lie algebra of $G$. The necessary condition for the universality of $\mc{S}$ is that the only matrices that commute with all $\mathrm{Ad}_{g_i}$'s are proportional to the identity. If in addition there is an element in $<\mc{S}>$ whose Hilbert-Schmidt distance from the centre of $G$ belongs to $]0,\frac{1}{\sqrt{2}}]$, then $\mc{S}$ is universal. Using these we provide a simple algorithm that allows deciding the universality of any set of $d$-dimensional gates in a finite number of steps and formulate a general classification theorem.\end{abstract}

\maketitle

\global\long\global\long\global\long\def\bra#1{\mbox{\ensuremath{\langle#1|}}}
 \global\long\global\long\global\long\def\ket#1{\mbox{\ensuremath{|#1\rangle}}}
 \global\long\global\long\global\long\def\bk#1#2{\mbox{\ensuremath{\ensuremath{\langle#1|#2\rangle}}}}
 \global\long\global\long\global\long\def\kb#1#2{\mbox{\ensuremath{\ensuremath{\ensuremath{|#1\rangle\!\langle#2|}}}}}
 
\section{Introduction}
Quantum computer is a device that operates on a finite dimensional quantum system $\mc{H}=\mc{H}_1\otimes\ldots \otimes \mc{H}_n$ consisting of $n$ {\it qudits} \cite{barenco,deutsch,Loyd} that are described by $d$-dimensional Hilbert spaces, $\mc{H}_i\simeq \mathbb{C}^d$ \cite{nielsen}. When $d=2$ qudits are typically called qubits. The ability to effectively manufacture optical gates operating on many modes, using for example optical networks that couple modes of light \cite{exp1,exp2,Reck}, is a natural motivation to consider not only qubits but also higher dimensional systems in the quantum computation setting (see also \cite{oszman, oszman2} for the case of fermionic linear optics and quantum metrology). One of the necessary ingredients for a quantum computer to work properly is the ability to perform arbitrary unitary operation on the system $\mc{H}$. We distinguish two types of operations. The first are one-qudit operations (one-qudit gates) that belong to $SU(\mc{H}_i)\simeq SU(d)$ and act on a single qudit. The second are $k$-qudit operations ($k$-qudit gates), $k\geq 2$, that belong to $SU(\mc{H}_{i_1}\otimes\ldots\otimes \mc{H}_{i_k})\simeq SU(d^k)$ and act on the chosen $k$ qudits. A $k$-qudit gate is nontrivial if it is not a tensor product of $k$ single qudit gates. We say that one-qudit gates $\mathcal{S}=\{g_1,\ldots,g_n\}$ are {\it universal} if any gate from $SU(d)$ can be built, with an arbitrary precision, using gates from $\mathcal{S}$. Mathematically this means that the set $<\mc{S}>$ generated by elements from $\mc{S}$ is dense in $SU(d)$ and its closure is the whole $SU(d)$, i.e. $\overline{<\mc{S}>}=SU(d)$. It is known that once we have access to a universal set of one-qudit gates together with one additional two-qudit gate that does not map separable states onto separable states, we can build within a given precision, an arbitrary unitary gate belonging to $SU(\mc{H})$ \cite{Brylinski} (see \cite{OZ17} for the similar criteria for fermionic and bosonic quantum computing). Thus in order to characterise universal sets of gates for quantum computing with qudits, one needs to characterise sets that are universal for one qudit. 

Although there are some qualitative characterisations of universal one-qudit gates, the full understanding is far from complete. It is known, for example, that almost all sets of qudit gates are universal, i.e  universal sets $\mathcal{S}$ of the given cardinality $c$ form a Zariski open set in $SU(d)^{\times c}$. By the definition of a Zariski open set we can therefore deduce that non-universal gates can be characterised by vanishing of a finite number of polynomials in the gates entries and their conjugates \cite{field, kuranishi}. These polynomials are, however, not known and it is hard to find operationally simple criteria that decide one-qudit gates universality. Some special cases of two and three dimensional gates  have been studied in \cite{aronson,sawicki1}. The main obstruction in these approaches is the lack of classification of finite and infinite disconnected subgroups of $SU(d)$ for $d>4$. Recently there were also approaches providing algorithms for deciding universality of a given set of quantum gates that can be implemented on quantum automatas \cite{Derksen05}.

The goal of this paper is to provide some simple criteria for universality of one-qudit gates that can be applied even if one does not know classification of finite/infinite disconnected subgroups of $SU(d)$. To achieve this we divide the problem into two. First, using the fact that considered gates $\mc{S}=\{g_1,\ldots,g_n\}$ belong to groups that are compact simple Lie groups $G$, we provide a criterion which allows to decide if an infinite subgroup is the whole group $G$. It is formulated in terms of the adjoint representation matrices $\mr{Ad}_g$, $g\in\mc{S}$ and boils down to finding the dimension of the commutant of all $\mr{Ad}_{g_i}$'s. The necessary condition for universality is that the commutant is one-dimensional. Checking this reduces to calculating the dimension of the kernel of a matrix constructed from $\mr{Ad}_{g_i}$'s, whose coefficients are polynomial in the entries of gates and their complex conjugates. Next, we give sufficient conditions for a set generated by $\mc{S}$ to be infinite. They stem from inequalities that relate the distances of two group elements and their commutators from the identity \cite{curtis,bottcher}. In particular we show that for a pair of gates $g_1$ and $g_2$, for which the Hilbert-Schmidt distances from the centre $Z(G)$ of $G$ are less than $\frac{1}{\sqrt{2}}$ and such that $[g_1,g_2]_\bullet:=g_1g_2g_1^{-1}g_2^{-1}\notin Z(G)$, deciding universality boils down to checking if the corresponding Lie algebra elements generate the whole Lie algebra. Next we show that for a gate whose distance from $Z(G)$ is larger that $\frac{1}{\sqrt{2}}$, $\mr{dist}(g,Z(G))\geq\frac{1}{\sqrt{2}}$, there is always $n\in\mathbb{N}$ such that $\mr{dist}(g^n,Z(G))<\frac{1}{\sqrt{2}}$. Moreover, using Dirichlet approximation theorems (and their modifications) we give an upper bound for the maximal $N_G$ such, that for every $g\in G$ we have $\mr{dist}(g^n,Z(G))<\frac{1}{\sqrt{2}}$ for some $1\leq n\leq N_G$. For the gates that satisfy the necessary condition for universality, we show that the group generated by $\mc{S}$ is either 1) finite iff the distance of all its elements (beside those belonging to $Z(G)$) from $Z(G)$ is longer than $\frac{1}{\sqrt{2}}$ or 2) otherwise equal to $G$. This key observation gives rise to a simple algorithm that allows to decide universality of any given set of gates. Moreover, it leads to a general classification theorem. In order to formulate it we  introduce the notion of the {\it exceptional spectrum}. For example, the spectrum of $g\in SU(d)$ is exceptional iff it is a collection of $n^{\mathrm{th}}$ roots of $\alpha\in \mathbb{C}$, where $1\leq n\leq N_{SU(d)}$ and $\alpha^d=1$. Notably there are only finitely many  exceptional spectra and their number can be easily calculated. Our classification theorem states that $\mc{S}$ which satisfies the necessary universality condition and contains at least one matrix with a non-exceptional spectrum is universal. Our approach for checking if the generated group is infinite is somehow related to \cite{freedman, jeandel}, however the conceptual differences in both approaches are significant and the methods should be treated as independent. The problem of deciding if a finitely generated group is infinite has been also studied and there are some algorithms that allow checking this property (see for example \cite{babai1, babai2,babai3,Derksen05}). In contrast to these approaches, our reasoning is based on the set of basic properties of compact connected simple Lie groups. The advantage for us of this approach is that it is explicit and direct. Moreover, the resulting algorithm is simple and can be easily implemented.

It is worth stressing here that universality criteria on the level of Lie groups require some additional conditions comparing to the level of Lie algebras. As an example, it was shown in \cite{schuch} that for the system of $n$ qubits, the set $\mc{S}$ consisting of all $1$-qubit gates and the SWAP gates between all pairs of qubits is not universal, whereas an analogous set of gates with the square roots of SWAP is universal. It is, however, evident that in both cases the corresponding Hamiltonians generate $\mk{su}(2^n)$. The interesting universal and non-universal extensions of local unitary gates in the setting for fermionic and bosonic quantum computing  can be also found in \cite{OZ17}.

In our paper we also demonstrate that the adjoint representation, this time for Lie algebras, can be useful in deciding if a finite subset $\mc{X}$ of a real compact semisimple Lie algebra generates the whole algebra (section \ref{algebra-gen}). This problem has been studied intensively in control theory \cite{albertini,brockett,jurdjevic} and in connection to universality of Hamiltonians, symmetries and controllability of quantum systems  \cite{Laura,Schirmer,zeier,daniel}. There are numerous criteria known and admittedly some are very general. Nevertheless, in section 3.1 we provide criteria for the universality of $\mc{X}$ using our approach with the adjoint representation.
As the considered groups are compact and connected, any gate $g\in G$ can be written as $g=e^X$, where $X$ is an element of the Lie algebra of the group. In Theorem \ref{thm-alg} we show that  when all elements $g\in \mc{S}$ satisfy $\mr{dist}(g,Z(G))<\frac{1}{\sqrt{2}}$  the necessary and sufficient condition for universality of $\mc{S}$ is completely  determined by generation of the Lie algebra  by the logarithms of the gates  from $\mc{S}$ (see Section \ref{log-construct} for the definition of the logarithm).

The last part of the paper concerns applications of the above ideas to $SU(2)$, $SO(3)$ and $SU(3)$. In particular we give a full characterisation of the universal pairs of single qubit gates and show that for any pair of $SU(2)$ gates our algorithm terminates for a word of the length $l\leq13$. Moreover, if the universality algorithm does not terminate in Step 2. with $1\leq l\leq 4$ the set $\mc{S}$ cannot be universal. We also show that for $SU(2)$ the exceptional spectra are in direct correspondence with the characters of the finite subgroups of $SU(2)$. We also characterise real and complex $2$-mode beamsplitters that are universal when acting on $d\geq 3$ modes. Our approach allows to reproof the results of \cite{aronson,sawicki1} without the knowledge of disconnected  infinite or finite subgroups of $SO(3)$ and $SU(3)$.

\section{Preliminaries}

\subsection{Compact semisimple Lie algebras}
A real Lie algebra is a finite dimensional vector space $\mathfrak{g}$ over $\RR$ together with a commutator $[\cdot,\cdot]:\mathfrak{g}\times\mathfrak{g}\rightarrow\mathfrak{g}$ that is: (1) bilinear (2) antisymmetric and (3) satisfies Jacobi identity $\left[\left[X,Y\right],Z\right]+\left[\left[Z,X\right],Y\right]+\left[\left[Y,Z\right],X\right]=0$.  In this paper we will often skip `real' as we will consider only real Lie algebras. A Lie algebra $\mathfrak{g}$ is nonabelian if there is a pair $X,Y\in\mathfrak{g}$ such that $[X,Y]\neq 0$. A subspace $\mathfrak{h}\subset\mathfrak{g}$ is a subalgebra of $\mathfrak{g}$ if and only if for any $X,Y\in\mathfrak{h}$ we have $[X,Y]\in\mathfrak{h}$, i.e. $\mathfrak{h}$ is closed under taking commutators. An important class of subalgebras are {\it ideals}. A subalgebra $\mathfrak{h}\subset\mathfrak{g}$ is an ideal of $\mathfrak{g}$ if for any $X\in \mathfrak{g}$ and any $Y\in \mathfrak{h}$ we have $[X,Y]\in\mathfrak{h}$. One easily checks that an intersection of ideals is an ideal. 
\begin{definition}
A nonabelian Lie algebra $\mathfrak{g}$ is simple if $\mk{g}$ has no ideals other than ${0}$ and $\mathfrak{g}$.
\end{definition}
We say that a Lie algebra $\mk{g}$ is a direct sum of Lie algebras, $\mk{g}=\oplus_{i=1}^n\mk{g}_i$, if and only if it is a direct sum of vector spaces  $\{\mk{g}_i\}_{i=1}^n$ and $[\mk{g}_i,\mk{g}_j]=0$ for all $i\neq j$. In this case $\mk{g}_i$'s are ideals of $\mk{g}$. The algebras we will be interested in belong to a special class of either simple Lie algebras or their direct sums. In the following we briefly discuss their properties.

A representation of a real Lie algebra on a real vector space is a linear map $\phi:\mathfrak{g}\rightarrow\mathrm{End}_\RR(V)$ that satisfies $\phi\left(\left[X,Y\right]\right)=\left[\phi(X),\phi(Y)\right]$. A representation is called {\it irreducible} if $V$ has no $\phi(\mk{g})$-invariant subspace $W\subset V$, i.e. a subspace for which $\phi(X)W\subset W$, for all $X\in\mk{g}$.

As $\mathfrak{g}$ is a real vector space itself, one can consider representation of $\mathfrak{g}$ on $\mathfrak{g}$. In fact, there exists a canonical representation of this type that is called the {\it adjoint representation}:
\begin{gather}
\mathrm{ad}:\mathfrak{g}\rightarrow\mathrm{End}(\mathfrak{g}),\,\,\mathrm{ad}_{X}(Y):=[X,Y].
\end{gather}
Note that invariant spaces of the adjoint representation are ideals and therefore the adjoint representation  of a simple Lie algebra is irreducible. Using the adjoint representation we define a bilinear form on $\mathfrak{g}$, called the {\it Killing form} given by $B(X,Y)=\mr{tr}\left(\mr{ad}_{X}\circ\mr{ad}_{Y}\right)$ \footnote{Upon a choice of basis in $\mathfrak{g}$ endomorphisms $\mr{ad}_{X}$ and $\mr{ad}_{Y}$ are matrices and hence we can compute the trace.}. The Killing form satisfies 
\begin{gather}\label{killing}
B\left(\mathrm{ad}_{X}\left(Y\right),Z\right)+B\left(\mathrm{ad}_{X}\left(Z\right),Y\right)=0.
\end{gather}
\begin{definition}
A real Lie algebra $\mathfrak{g}$ is a \underline{compact semisimple Lie algebra} if its Killing form is negative definite.
\end{definition}
Assume now that $\mathfrak{g}$ is a compact semisimple Lie algebra and let $\mathfrak{a}\subset\mk{g}$ be an ideal. Let $\mathfrak{a}^\perp$ be the orthogonal complement of $\mathfrak{a}$ with respect to the Killing form. For any $X\in\mathfrak{g}$, $Y\in\mathfrak{a}^\perp$, and $Z\in\mathfrak{a}$ we have
\begin{gather}
B\left(\left[X,Y\right],Z\right)=-B\left(Y,\left[X,Z\right]\right)=0.
\end{gather}
Hence $[X,Y]\in\mk{a}^\perp$. Therefore $\mathfrak{a}^\perp$ is also an ideal. Note next that $[\mk{a},\mk{a}^\perp]\subset\mk{a}\cap\mk{a}^\perp$. The restriction of $B$ to the ideal $\mk{a}\cap\mk{a}^\perp$ is obviously zero. But $B$ is negative definite, hence $\mk{a}\cap\mk{a}^\perp=0$. As a result $\mathfrak{g}=\mathfrak{a}\oplus\mathfrak{a}^\perp$ is a direct sum of ideals. We can repeat this procedure for $\mk{a}$ and $\mk{a}^\perp$ and after a finite number of steps finally we get: 
\begin{fact}
A real compact semisimple Lie algebra is a direct sum of real compact simple Lie algebras.
\end{fact}
Let us next choose a basis $\left\{X_i\right\}_{i=1}^{\mr{dim}\mk{g}}$ in $\mk{g}$ that satisfies $B(X_i,X_j)=-\delta_{ij}$. In this basis $\mathrm{ad}_X$ is an antisymmetric trace zero real matrix, hence an element of the special orthogonal Lie algebra $\mk{so}(\mr{dim}\;\mk{g})$. Finally we remark that the subalgebra of a simple or a semisimple Lie algebra need not to be simple/semisimple. 

\subsection{Compact semisimple Lie groups}
A Lie group $G$ is a group that has a structure of a differential manifold and the group operations are smooth. We say that $G$ is compact if it is a compact manifold, i.e. any open covering of $G$ has a finite subcovering. It is well known that a closed subgroup of a Lie group is a Lie group \cite{nojman,Cartan}. In this section we will always consider closed subgroups. An important class of subgroups are normal subgroups. $H\subset G$ is a normal subgroup if for each $g\in G$ we have $gHg^{-1}\subset H$. We denote it by $H\triangleleft G$. In this case the quotient $G/H$ is a group. A disconnected $G$ consists of connected components. Connected components of a Lie group are open and their number is finite if $G$ is compact, as otherwise they would constitute an open covering of $G$ that does not possess finite subcovering.  The identity component $G_e$, i.e. the component that contains the neutral element $e$, is a normal subgroup of $G$. This can be easily seen as the maps $\phi_g:G\rightarrow G$, $\phi_g(h)=ghg^{-1}$ are continuous for every $g\in G$, hence they map components into components. But $e\in \phi_g(G_e)$ for all $g\in G$, hence $ \phi_g(G_e)=G_e$. The quotient $G/G_e$ is a group (because $G_e$ is normal) which for a compact $G$ is a finite group called the {\it components group}. 

The connection between Lie groups and Lie algebras is established in the following way. Left invariant vector fields on $G$ together with vector fields commutators form the Lie algebra $\mk{g}$ of a Lie group $G$. Note that these fields are determined by their value at $e$ and therefore $\mk{g}$ can be identified with the tangent space to $G$ at $e$, i.e $\mathfrak{g}=T_eG$. For every $X\in\mk{g}$ there is a unique one parameter subgroup $\gamma(t)$ whose tangent vector at $e$ is $X$. We define the exponential map $\exp:
\mathfrak{g}\rightarrow G$ to be: $\exp(X):=\gamma(1)$. For any Lie group the image of the exponential map, $\exp(\mk{g})$, is contained in the  identity component $G_e$ and when $G$ is compact $\exp(\mathfrak{g})=G_e$. Therefore for a compact and connected group every element $g\in G$ is of the form $\exp(X)$ for some $X\in \mk{g}$. For matrix Lie groups $G\subset\mathrm{GL}(n,\mathbb{C})$ these definitions simplify as the exponential map is the matrix exponential that is defined by $e^X=\sum_{i=0}^\infty\frac{X^n}{n!}$ and the Lie algebra is defined as $\mk{g}=\{X:e^{tX}\in G,\,\,\forall t\in \RR\}$.

\begin{definition}
A compact connected Lie group is simple/semisimple if its Lie algebra is a compact simple/ a compact semisimple Lie algebra.
\end{definition}

Recall that the  Lie algebra $\mk{h}$ of the identity component of $H\triangleleft G$ is an ideal of the Lie algebra $\mathfrak{g}$. We can also use an equivalent definition that says a compact connected group $G$ is simple if it has no connected normal subgroups. Similarly as for Lie algebras, compact semisimple Lie groups have a particularly nice structure. 
\begin{fact}
Let G be a compact connected semisimple group. Then $$G=\left(G_1\times\ldots\times G_k\right)/Z,$$ where each $G_i$ is a simple compact group and $Z$ is contained in the centre of $G_1\times\ldots\times G_k$.
\end{fact}
A representation of a Lie group on a real vector space is a homomorphism $\Phi:G\rightarrow\mathrm{GL}_\RR(V)$, i.e. $\Phi$ satisfies $\Phi(g_1g_2)=\Phi(g_1)\Phi(g_2)$. A particularly important example is the adjoint representation of $G$ on $\mathfrak{g}$.
\begin{gather}\label{ad-rep}
\mathrm{Ad}:G\rightarrow\mathrm{Aut}(\mathfrak{g}),\,\,\mathrm{Ad}_g(X):=gXg^{-1}.
\end{gather}
The image of $\mathrm{Ad}_G$ is $\mathrm{Ad}_G=G/Z(G)$, where $Z(G)$ is the centre of $G$. For a semisimple compact Lie group $Z(G)$ is finite  by definition and therefore $\mathrm{Ad}$ is a finite covering homomorphism onto $G/Z(G)$. For a compact connected simple Lie groups the adjoint representation is irreducible.

The relation between the adjoint representations of a compact connected semisimple Lie group and its Lie algebra, $\mr{Ad}$ and $\mr{ad}$, follows from the fact that $\mr{Ad}$ is a smooth homomorphism. For any $X\in\mk{g}$ and all $t\in\RR$ elements $\mr{Ad}_{e^{tX}}$ form a one-parameter subgroup in $\mathrm{Aut}(\mathfrak{g})$ whose tangent vector at $t=0$ is $\mr{ad}_X$. As this group is uniquely determined by its tangent vector we have $\mathrm{Ad}_{e^{tX}}=e^{\mr{ad}_{tX}}$. Using this relation we easily see that the Killing form on $\mk{g}$ is invariant with respect to the adjoint action, i.e  $B(\mr{Ad}_g X,\mr{Ad}_g Y)=B(X,Y)$. Recall that for a compact semisimple $G$ the Killing form is an inner product (negative definite) and therefore $\mr{Ad}_g$ is an orthogonal matrix belonging to $SO(\mathfrak{g})$. After the choice of an orthonormal basis in $\mathfrak{g}$, using (\ref{ad-rep})  we can calculate entries of the matrix $\mr{Ad}_g$. It is easy to see that this  matrix belongs to $SO(\mr{dim}\;\mathfrak{g})$.

\subsection{Subgroups of a compact semisimple Lie group}\label{sec:sub}
Let $G$ be a Lie group. We say that $H\subset G$ is a discrete subgroup of $G$ if there is an open cover of $H$ such that every open set in this cover contains exactly one element from $H$ - we will call it a {\it discrete open cover} of $H$. If $G$ is compact every discrete subgroup is finite. To see this, assume that there is an infinite discrete subgroup $H$ in a compact $G$ and take the open cover of $G$ that is a union of the discrete open cover of $H$ and the open set which consists of elements not in this discrete cover. Then this cover is infinite and has no finite subcover, hence we get contradiction. By the similar argument any closed disconnected subgroup $H$ of a compact $G$ has finitely many connected components. The Lie algebra $\mk{h}$ of the identity component $H_e$ is a subalgebra of $\mk{g}$ and the exponential map is surjective onto $H_e$, however $\mathfrak{h}$ needs not to be semisimple. We distinguish three possible types of closed subgroups of the compact Lie group $G$: (1) finite discreet subgroup\textcolor{red}{s}, (2) disconnected subgroups with a finite number of connected components, (3) connected subgroups.

In this paper we consider groups that are generated by a finite number of elements from some compact semisimple Lie group $G$. More precisely for $\mc{S}=\{g_1,\ldots,g_k\}\subset G$ we consider the closure of 
\[
<\mc{S}>:=\left\{g_{i_1}^{k_1}\cdot\ldots\cdot g_{i_m}^{k_m}:g_{i_j}\in \mathcal{S},k_j\in \mathbb{N}, i_j\in\{1,\ldots,n\}\right\},
\]
which is a Lie subgroup of $G$ (see Fact \ref{liesub} for the proof). In particular we want to know when $\overline{<\mc{S}>}=G$. It is known that almost any two elements of $G$ generate a compact semisimple $G$. Moreover, as was shown by Kuranishi \cite{kuranishi} elements that are in a sufficiently small neighbourhood of $e$ generate $G$ if and only if their corresponding Lie algebra elements generate $\mathfrak{g}$. The proof is, however, not constructive. The author of \cite{field} shows that pairs generating $G$ form a Zariski open subset of $G\times G$. In our work we adopt and develop some of the ideas contained in \cite{kuranishi} and \cite{field} and this way obtain characterisation of sets $\mc{S}$ that generate groups $SU(d)$ or $SO(d)$. Moreover, our approach results with a simple algorithm that enables deciding the universality of any given set of gates. For the completeness we prove the following.
\begin{fact}\label{liesub}
The closure of $<\mc{S}>$ is a Lie group.
\end{fact}
\begin{proof}
By the theorem of Cartan \cite{Cartan,nojman} we know that a closed subgroup of a Lie group is a Lie group. The set $\overline{\mc<S>}$ is obviously closed and hence we are left with showing that it is has a group structure. By the construction $\mc{S}$ is invariant under multiplication and therefore $\overline{<\mc{S}>}$ has this property too. As a direct implication of Dirichlet approximation theorem (see theorem \ref{dn}), for every element $g\in\mc{S}$ there is a sequence $\{g^{n_k}\}$, such that $g^{n_k}\rightarrow I$ when $k\rightarrow \infty$. Thus $I\in \overline{<\mc{S}>}$. Note, however, that by the same argument the sequence $\{g^{n_k-1}\}\subset \mc{S}$ converges to $g^{-1}$. Thus $\overline{<\mc{S}>}$ has a group structure. The result follows.   

\end{proof}

In order to clarify the terminology, whenever we say the group generated by $\mc{S}$ we mean the compact Lie group $\overline{<\mc{S}>}$.
\section{Generating sets for compact semisimple Lie algebras and Lie groups}
We begin with some remarks concerning irreducible representations on real vector spaces that we will call irreducible real representations. The well known version of the Schur lemma states that a representation of a Lie group or a Lie algebra on a complex vector space (complex representation) is irreducible iff the only matrices that commute with all representation matrices are $\{\lambda I: \lambda \in \mathbb{C}\}$. In our paper the considered representations are irreducible real representations. A real irreducible representation can be of (1) real type, (2) complex type, or (3) quaternion type. The type of representation determines the structure of endomorphisms commuting with the representation matrices (see chapter II.6 of \cite{Dieck} for full discussion). The following theorem holds (theorem II.6.7 of \cite{Dieck})

\begin{fact}(Schur Lemma) 
For a real irreducible representation of (1) real, (2) complex, (3) quaternion type the algebra of endomorphisms commuting with the representation matrices if isomorphic to (1) $\mathbb{R}$, (2) $\mathbb{C}$, (3) $\mathbb{H}$, respectively, where $\mathbb{H}$ stands for Hamilton quaternions. 
\end{fact}

Next we show that the adjoint representation for a compact simple Lie group/algebra is of the real type. Using Table II.6.2 and Propositions II.6.3 of \cite{Dieck} it suffices to show that its complexification is of the real type. On the other hand, by Proposition II.6.4 it reduces to showing that the complexfication $\mk{g}^\mathbb{C}$ of a compact simple Lie algebra $\mathfrak{g}$  posseses a symmetric, non-degenerate and $\mr{Ad}_G$-invariant form. To this end we define the Killing form on $\mathfrak{g}^\mathbb{C}$ in the analogues way as in $\mathfrak{g}$, i.e. $B_{\mk{g}^\mathbb{C}}(X_1,X_2)=\mathrm{tr}(\mathrm{ad}_{X}\circ \mathrm{ad}_{Y})$, $X,Y\in\mathfrak{g}^\mathbb{C}$. Note that a basis of $\mathfrak{g}$ over $\mathbb{R}$ is a basis of $\mathfrak{g}^\mathbb{C}$ over $\mathbb{C}$. Thus $B_{\mk{g}^\mathbb{C}}$ is a non-degenerate symmetric $\mathrm{Ad}_G$-invariant form as the Killing form for $\mk{g}$ is such. Hence:
\[
\mc{C}(\mr{ad}_\mk{g})=\{\lambda I: \lambda \in \mathbb{R}\}=\mathcal{C}({\mr{Ad}_G}).
\]

\subsection{Generating sets for compact semisimple Lie algebras}\label{algebra-gen}
In this section $\mathfrak{g}$ will denote a compact semisimple Lie algebra. Let $\mathcal{X}=\{X_1,\ldots,X_n\}\subset \mathfrak{g}$. We say that $\mc{X}$ generates $\mk{g}$ if any element of $\mk{g}$ can be written as a finite linear combination of $X_i$'s and finitely nested commutators of $X_i$'s:
\[
\sum_{i}\alpha_iX_i+\sum_{i,j}\alpha_{i,j}[X_i,X_j]+\ldots.
\]
Our aim is to provide a general criterion that uses the adjoint representation of compact semisimple Lie algebras to verify when $\mc{X}\subset\mk{g}$ generates $\mk{g}$. This problem has been studied over the years and there are many other approaches that do not use the adjoint representation. It is also  an important question in to control theory as it plays central role in controllability of certain dynamical systems \cite{albertini,brockett,jurdjevic}. The corresponding conditions are known as the so-called {\it Lie algebra rank condition} \cite{brockett,jurdjevic}.The more recent conditions that are in the spirit of what we will present in Lemma \ref{ad2} include \cite{zeier1, zeier} and in particular \cite{daniel} where the problem for compact Lie algebras is studied. As we will see in the next section conditions for generation of Lie algebras are too weak when one considers generation of Lie groups. Thus this section plays a marginal role for the rest of the paper (excluding Theorem \ref{thm-alg}). The main purpose of this section is to give evidence that the adjoint representation can be useful in deciding both Lie algebras and Lie groups generation problem.

Let $\mc{C}(\mr{ad}_\mk{g})=\left\{L\in\mr{End}(\mk{g}): \forall X\in\mk{g}\,\left[\mr{ad}_X,L\right]=0\right\}$ denotes the space of endomorphisms of $\mathfrak{g}$ that commute with all $\mr{ad}_X$, $X\in\mk{g}$. By the Jacobi identity $\mc{C}(\mr{ad}_\mk{g})$ is a Lie subalgebra of $\mr{End}(\mk{g})$. Moreover, also by Jacobi identity, if $L\in\mr{End}(\mk{g})$ commutes with $\mr{ad}_X$ and $\mr{ad}_Y$ then it also commutes with $\mr{ad}_{\alpha X+\beta Y}$ and $\mr{ad}_{[X,Y]}$. Let us denote by $\mc{C}(\mr{ad}_{\mc{X}})$ the solution set of 
\[
\left[\mr{ad}_{X_1},\cdot\right]=0,\ldots,\left[\mr{ad}_{X_n},\cdot\right]=0.
\]
It is clear that if $\mc{X}$ generates $\mk{g}$, then $\mc{C}(\mr{ad}_\mk{g})=\mc{C}(\mr{ad}_\mc{X})$. It happens that the converse is true for semisimple Lie algebras. Let next $\mk{g}=\mk{g}_1\oplus\ldots\oplus\mk{g}_k$ be a decomposition of a semisimple $\mk{g}$ into simple ideals. Let $\mc{X}=\{X_1,\ldots,X_n\}\subset\mk{g}$. Every $X_i\in\mc{X}$ has a unique decomposition:
\[
X_i=X_{i,1}+\ldots + X_{i,k}, \,\mr{where}\,X_{i,j}\in\mk{g}_j.
\]
Therefore $\mc{X}$ generates $\mk{g}$ if every set $\mc{X}_i=\{X_{1,i},\ldots,X_{n,i}\}$ generates $\mk{g}_i$, $i\in\{1,\ldots,k\}$. Note that if the projection of $\mc{X}$ onto some simple component of $\mk{g}$ is zero than $\mc{X}$ cannot generate and  $\mk{g}$ and $\mc{C}(\mr{ad}_\mk{g})\neq \mc{C}(\mr{ad}_\mc{X})$. Thus the equality  $\mc{C}(\mr{ad}_\mk{g})=\mc{C}(\mr{ad}_\mc{X})$ implies that $\mc{X}$ has nonzero intersection with every simple component of $\mk{g}$.
\begin{lem}\label{ad2}
Let $\mathfrak{g}$ be a compact semisimple Lie algebra and $\mathcal{X}=\{X_1,\ldots,X_n\}\subset \mathfrak{g}$ its finite subset. $\mc{X}$ generates $\mk{g}$ \textbf{if and only if} $\mc{C}(\mr{ad}_\mk{g})=\mc{C}(\mr{ad}_\mc{X})$.
\end{lem}
\begin{proof}
Let $n$ be the number of components of $\mk{g}$ and let us denote by $\mk{h}\subset\mk{g}$ the Lie algebra generated by $\mc{X}$. Assume that $\mk{h}\neq\mk{g}$ but $\mc{C}(\mr{ad}_\mk{g})=\mc{C}(\mr{ad}_\mc{X})$. The equality of commutants implies that $\mk{h}$ has nonzero intersection with every simple component of $\mk{g}$. Using the Killing form we can decompose $\mk{g}$ into a direct product of vector spaces (not necessarily Lie algebras), $\mk{g}=\mk{h}\oplus\mk{h}^\perp$. For any $X\in\mk{h}$, $Y\in \mk{h}$ and $Z\in\mk{h}^\perp$ we have $\mr{ad}_XY\in\mk{h}$ and $\mr{ad}_XZ\in \mr{h}^\perp$. The latter is true as $B(\mr{ad}_XZ,Y)=-B(Z,\mr{ad}_XY)=0$, for any $Y\in\mk{h}$. Therefore, for $X\in\mk{h}$ operators $\mr{ad}_X$ respect the decomposition $\mk{g}=\mk{h}\oplus\mk{h}^\perp$ and have a block diagonal structure:
\begin{gather}
\mr{ad}_X=\left(\begin{array}{cc}
\mr{ad}_{X}\big |_{\mk{h}}&0\\ 
0&\mr{ad}_{X}\big |_{\mk{h}^\perp}\end{array} \right).
\end{gather}
Let $P:\mk{g}\rightarrow\mk{h}$ be the orthogonal, with respect to the Killing form, projection operator onto $\mk{h}$. Then obviously $[P,\mr{ad}_X]=0$ for any $X\in\mk{h}$. Note, however, that if $P$ belonged to $\mc{C}(\mr{ad}_\mk{g})$ then $\mk{h}$ would be an ideal of $\mk{g}$. But the only ideals of $\mk{g}$ are direct sums of its simple components. Thus $\mk{h}$ is either  $\mk{g}$ which is a contradiction or $\mk{h}$ is a direct sum of $k< n$ simple components of $\mk{g}$ which is again a contradiction.
\end{proof}

Using the Schur lemma we obtain:
\begin{coll}\label{ad1}
Let $\mathfrak{g}$ be a compact simple Lie algebra and $\mathcal{X}=\{X_1,\ldots,X_n\}\subset \mathfrak{g}$ be its finite subset. $\mc{X}$ generates $\mk{g}$ \textbf{if and only if} $\mc{C}(\mr{ad}_\mk{g})=\{\lambda I:\lambda\in\RR\}$.
\end{coll}

Finally let us remark that it is very important to consider not a defining but the adjoint representation. To see this let $X_1,X_2$ be two matrices that generate $\mk{su}(2)$ and consider the set $\mc{X}=\left\{X_1\otimes I,X_2\otimes I, I\otimes X_1,I\otimes X_2\right\}\subset \mk{su}(4)$. Note that the Lie algebra generated by $\mc{X}$ is $\mathfrak{su}(2)\oplus\mk{su}(2)\subset\mk{su}(4)$. One checks by direct calculations that the only $4\times 4$ matrix commuting with $\mc{X}$ is proportional to the identity. This is, however, not the case for matrices $\mr{ad}_X$, $X\in\mc{X}$. Hance changing the adjoint representation in Corollary \ref{ad1} into the defining one would result in the equality between $\mk{su}(2)\oplus\mk{su}(2)$ and $\mk{su}(4)$ which is of course not true.

\subsection{Generating sets for compact semisimple Lie groups}

We are interested in the the following problem. Let $G$ be a compact connected semisimple Lie group and let $\mathcal{S}=\left\{g_1,\ldots,g_n\right\}\subset G$. We want to know when $\overline{<\mc{S}>}=G$. To this end we use adjoint representation. 


Let $\mc{C}(\mr{Ad}_G)=\left\{L\in\mr{End}(\mk{g}): \forall g\in G\,\left[\mr{Ad}_g,L\right]=0\right\}$ denote the space of endomorphisms of $\mathfrak{g}$ that commute with all $\mr{Ad}_g$, $g\in G$. By the Jacobi identity $\mc{C}(\mr{Ad}_G)$ is a Lie subalgebra of $\mr{End}(\mk{g})$. Moreover, if $L\in\mr{End}(\mk{g})$ commutes with $\mr{Ad}_g$ and $\mr{Ad}_h$ then it also commutes with $\mr{Ad}_{gh}$. Let us denote by $\mc{C}(\mr{Ad}_{\mc{S}})$ the solution set of 
\[
\left[\mr{Ad}_{g_1},\cdot\right]=0,\ldots,\left[\mr{Ad}_{g_n},\cdot\right]=0.
\]
It is clear that if $\mc{S}$ generates $G$ then $\mc{C}(\mr{Ad}_G)=\mc{C}(\mr{Ad}_\mc{S})$. It happens that with some additional assumptions the converse is true for semisimple Lie groups.
\begin{lem}\label{Ad2}
Let $G$ be a compact connected semisimple Lie group and $\mathcal{S}=\{g_1,\ldots,g_n\}\subset G$ its finite subset such that $<\mc{S}>$ is infinite and the projection of $<\mc{S}>$ onto every simple component of $G$ is also infinite. $\mc{S}$ generates $G$ \textbf{if and only if} $\mc{C}(\mr{Ad}_G)=\mc{C}(\mr{Ad}_\mc{S})$.
\end{lem}

\begin{proof}
Let us denote by $H$ the closure of the group generated by $\mc{S}$, i.e. $H=\overline{<\mc{S}>}$. $H$ is a compact Lie group that contains infinite number of elements. Let $H_e$ be the identity component of $H$. As we know $H_e$ is a normal subgroup of $H$. Let $\mk{h}\subset\mk{g}$ be the Lie algebra of $H_e$ and let $n$ be the number of simple components of $\mk{g}=Lie(G)$. Under our assumption $\mk{h}$ has nonzero intersection with every simple component of $\mk{g}$. Assume that $\mk{h}\neq\mk{g}$ but $\mc{C}(\mr{Ad}_G)=\mc{C}(\mr{Ad}_\mc{S})$. Using the Killing form we can decompose $\mk{g}$ into a direct product of vector spaces (not necessarily Lie algebras), $\mk{g}=\mk{h}\oplus\mk{h}^\perp$. For any $g\in H$, $X\in \mk{h}$ and $Y\in\mk{h}^\perp$ we have $\mr{Ad}_gY\in\mk{h}$ and $\mr{Ad}_gY\in \mr{h}^\perp$. The latter is true as $B(\mr{Ad}_gY,X)=B(Y,\mr{Ad}_{g^{-1}}X)=0$, for any $X\in\mk{h}$. Therefore, for $h\in H$ the operators $\mr{Ad}_h$ respect the decomposition $\mk{g}=\mk{h}\oplus\mk{h}^\perp$ and have a block diagonal structure:
\begin{gather}
\mr{Ad}_h=\left(\begin{array}{cc}
\mr{Ad}_{h}\big |_{\mk{h}}&0\\ 
0&\mr{Ad}_{h}\big |_{\mk{h}^\perp}\end{array} \right).
\end{gather}
Let $P:\mk{g}\rightarrow\mk{h}$ be the orthogonal projection with respect to the Killing form onto $\mk{h}$. Then obviously $[P,\mr{Ad}_h]=0$ for any $h\in H$. Note, however, that if $P$ belonged to $\mc{C}(\mr{Ad}_G)$ then $\mk{h}$ would be  $\mr{Ad}_G$ invariant subspace of $\mk{g}$. But the only $\mr{Ad}$-invariant subspaces of $\mathfrak{g}$ are simple components of $\mk{g}$. Hence either $\mathfrak{h}=\mk{g}$ which is a contradiction or $\mathfrak{h}$ is a direct sum of $k<n$ simple components of $\mk{g}$ which again is a contradiction as $\mk{h}$ has nonzero intersection with all $n$ simple components.
\end{proof}

Using the Schur lemma we obtain:
\begin{coll}\label{Ad1}
Let $G$ be a compact connected simple Lie group and $\mathcal{S}=\{g_1,\ldots,g_n\}$ its finite subset. Assume $<\mc{S}>$ is infinite. The set $\mc{S}$ generates $G$ if and only if $\mc{C}(\mr{Ad}_G)=\{\lambda I:\lambda\in\RR\}$.
\end{coll}

%

Finally, note that $\overline{<\mc{S}>}$ is infinite in particular when at least one of $g_i$'s is of infinite order. Hence:

\begin{coll}\label{col:ad_infOrder}
Let $G$ be a compact connected simple Lie group and $\mathcal{S}=\{g_1,\ldots,g_n\}\subset G$ its finite subset such that at least one of $g_i$'s is of infinite order. $\mc{S}$ generates $G$ if and only if $\mc{C}(\mr{Ad}_\mc{S})=\{\lambda I:\lambda\in\RR\}$.
\end{coll}
In the next section we characterise when $<\mc{S}>$ is infinite and when $\mc{C}(\mr{Ad}_\mc{S})$ can be different form $\mc{C}(\mr{ad}_\mc{X})$ for semisimple groups of our interest, i.e. for $G=SU(d)$ and $G=SO(d)$.

\section{Groups $SU(d)$ and $SO(d)$}
In this section we focus on two groups $G$ that are particularly important from the perspective of quantum computation and linear quantum optics, i.e. $G=SO(d)$ or $G=SU(d)$. 
\begin{gather}
SO(d)=\{O\in \mr{Gl}_d(\mathbb{R}):O^tO=I,\,\mr{det}O=1\},\\
SU(d)=\{U\in \mr{Gl}_d(\mathbb{C}):U^\dagger U=I,\,\mr{det}X=1\}.
\end{gather}
Their Lie algebras $\mk{g}$ are:
\begin{gather}
\mk{so}(d)=\{X\in \mr{Mat}_d(\mathbb{R}):X^t=-X,\,\mr{tr}X=0\},\\
\mk{su}(d)=\{X\in \mr{Mat}_d(\mathbb{C}):X^\dagger=-X,\,\mr{tr}X=0\}.
\end{gather}
The centres of $G$ are finite and given by $Z(SU\left(d\right))=\left\{\alpha I:\alpha\in\mathbb{C},\,\alpha^d=1\right\}$, $Z(SO\left(2d\right))=\{\pm I\}$ and  $Z(SO\left(2d+1\right))=I$. Groups $SU(d)$ for $d\geq 2$ and groups $SO(d)$ for $d\geq 3$ and $d\neq 4$ are compact connected simple Lie groups. On the other hand $SO(4)$ is still compact and connected but it is not simple as its Lie algebra is a direct sum of Lie algebras $\mk{so}(4)=\mk{so}(3)\oplus\mk{so}(3)$, hence $SO(4)$ is semisimple. The Killing form on both $\mk{su}(d)$ and $\mk{so}(d)$, up to a constant positive factor, is given by $B(X,Y)=\mr{tr}XY$. We next introduce an orthonormal basis in $\mk{su}(d)$ and $\mk{so}(d)$. Let $E_{kl}=\kb kl$ be a $d\times d$ matrix whose only nonzero (and  equal to $1$) entry is $(k,l)$. The commutation relations are $\left[E_{ij},E_{kl}\right]=\delta_{jk}E_{il}-\delta_{li}E_{k,j}$. Let
\begin{gather}\label{Xij}
X_{ij}=E_{ij}-E_{ji},\,\, Y_{ij}=i\left(E_{ij}+E_{ji}\right),\,\, Z_{ij}=i(E_{ii}-E_{jj}).
\end{gather}
One easily checks that for  $i,j\in\{1,\ldots,d\}$, $i<j$ matrices $\left\{X_{ij}, Y_{ij}, Z_{i,i+1}\right\}$ form an orthogonal basis of $\mk{su}(d)$ and matrices $\left\{X_{ij}\right\}$ of $\mk{so}(d)$. We will call these two bases the standard basis of $\mk{su}(d)$ and $\mk{so}(d)$ respectively. 

\subsection{Gates and their Lie algebra elements}\label{log-construct}
In this section we explain how to any set of gates $\mc{S}$ we assign the set of Lie algebra elements $\mc{X}$. 

Let us recall that for a unitary matrix $U\in SU(d)$ there is a unitary matrix $V\in SU(d)$ such that $D=V^\dagger UV=\mr{diag}\{e^{i\phi_1},\ldots, e^{i \phi_d}\}$. The nonzero entries of $D$ constitute the spectrum of $U$. In order to find $X\in \mathfrak{su}(d)$ such that $U=e^X$ one should calculate a logarithm of $U$. This can be done using the decomposition $U=VDV^\dagger$ and it boils down to calculating logarithms of diagonal matrix $D$. Since the logarithm of $z\in\mathbb{C}$ is not uniquely defined we will use the convention that $\log z=\mathrm{arg}(z)$, where $\mathrm{arg}(z)$ is the argument of $z$ and we assume $\mathrm{arg}(z)\in [0,2\pi)$. Thus we choose $X\in \mathfrak{su}(d)$ that satisfies $U=e^X$ as $X=V\tilde{D}V^\dagger$, where $\tilde{D}=\mr{diag}\{i\phi_1,\ldots, i \phi_d\}$, every $\phi_i\in [0,2\pi)$. This way to any set of gates $\mc{S}=\{U_1,\ldots,U_n\}\subset SU(d)$ we assign the set of Lie algebra elements $\mathcal{X}=\{X_1,\ldots, X_n\}\subset \mathfrak{su}(d)$.

Matrices in $SO(d)$ typically cannot be diagonalised by the orthogonal group. Nevertheless for a matrix $O\in SO(d)$ there is an orthogonal matrix $V$ such that $R=V^tOV$ is block diagonal with two types of blocks: (1) one identity matrix $I_k$ of dimension $0\leq k\leq d$, (2) $2\times 2$ rotations by angles $\phi_i\in (0,2\pi)$, i.e. matrices $O(\phi_i)$ from $SO(2)$. We again want to find $X\in \mk{so}(d)$ such that $O=e^X$. In our paper we choose $X=V\tilde{R}V^t$, where $\tilde{R}$ has the same block diagonal structure as $R$ and (1) the block of $\tilde{R}$ corresponding to the identity block of $R$ is the zero matrix $0_k$ of dimension $0\leq k\leq d$, (2) the blocks corresponding to $2\times 2$ $\phi_i$-rotation blocks of $R$ are matrices $\left(\begin{array}{cc}0&\phi_i\\-\phi_i&0\end{array}\right)\in\mk{so}(2)$, where every $\phi_i\in(0,2\pi)$. We will call $R$ and $\tilde{R}$ normal forms of $O\in SO(d)$ and $X\in \mk{so}(d)$ respectively and angles $\phi_i$'s the spectral angles. Summing up, using the above procedure, to any set of gates $\mc{S}=\{O_1,\ldots,O_n\}\subset SO(d)$ we assign the set of Lie algebra elements $\mathcal{X}=\{X_1,\ldots, X_n\}\subset \mathfrak{so}(d)$.

Throughout the paper, whenever we speak about the Lie algebra elements associated to gates (or the logarithms of the gates) we mean matrices constructed according to the above two procedures.

\subsection{The difference between $\mc{C}(\mr{Ad}_\mc{S})$ and $\mc{C}(\mr{ad}_\mc{X})$}
\subsubsection{The case of $SU(d)$}\label{sec:adAd__sud}
Let $\mc{S}=\{U_1,\ldots,U_n\}\subset SU(d)$ and let  $\mc{X}=\{X_1,\ldots,X_n\}$ be the corresponding set of Lie algebra elements (constructed as described in Section \ref{log-construct}). In this section we study when the spaces  $\mc{C}(\mr{Ad}_\mc{S})$ and $\mc{C}(\mr{ad}_\mc{X})$ are different. Note first that using $\mr{Ad}_{e^{X_i}}=e^{\mr{ad}_{X_i}}$ we have $\mc{C}(\mr{ad}_\mc{X})\subseteq\mc{C}(\mr{Ad}_\mc{S})$.  Hence we are particularly interested in the situation when $\mc{C}(\mr{Ad}_\mc{S})$ is strictly larger then $\mc{C}(\mr{ad}_\mc{X})$. Matrices $U_i$ can be put into diagonal form $U_i=V_iD_iV_i^\dagger$, where $V_i\in SU(d)$ and $D_i=\{e^{\phi^i_{1}},\ldots ,e^{\phi^i_{d}}\}$, $\phi^i_{j}\in [0,2\pi)$. Note now that $\mr{Ad}_{U_i}=\mr{Ad}_{V_iD_iV_i^\dagger}=O_i\mr{Ad}_{D_i}O_i^t$, where $O=\mr{Ad}_{V_i}\in SO(d^2-1)$. Let us order the standard basis of $\mk{su}(d)$ as follows $\{X_{12},Y_{12},\ldots,X_{d-1,d},Y_{d-1,d},Z_{1,2},\ldots Z_{d-1,d}\}$. The matrix $\mr{Ad}_{D_i}$ in this basis has a block diagonal form:
\begin{gather}
\mathrm{Ad}_{D_{i}}=\left(\begin{array}{ccccccccc}
O(\phi^i_{1,2})\\
 & \ddots\\
 &  & O(\phi^i_{1,d})\\
 &  &  &   \ddots\\
 &  &  &  &   O(\phi^i_{2,d})\\
 &  &  &  &  &   \ddots\\
 &  &  &  &  &  &   O(\phi^i_{d-1,d})\\
 &  &  &  &  &  &  &   I_{d-1}
\end{array}\right),
\end{gather}
where

\begin{gather}
O(\phi^i_{k,l})=\left(\begin{array}{cc}
\cos(\phi^i_{k,l})&\sin(\phi^i_{k,l})\\
-\sin(\phi^i_{k,l} )& \cos(\phi^i_{k,l})
\end{array}\right),\, \mr{where},\, \phi^i_{k,l}:=\phi^i_k-\phi^i_l,
\end{gather}
and $I_{d-1}$ is $(d-1)\times (d-1)$ identity matrix. Matrices from $\mc{X}$ are given by $X_i=V_i\tilde{D}_iV_i^\dagger$ and $\tilde{D}_i=i\{\phi^i_1,\phi^i_2,\ldots,\phi^i_d\}$. Hence $\mr{ad}_{X_i}=\mr{ad}_{V_i\tilde{D}_iV_i^\dagger}=O\mr{ad}_{\tilde{D}_i}O^t$, and we have (in the standard basis of $\mk{su}(d)$ ordered as previously):
\begin{gather}
\mathrm{ad}_{\tilde{D}_{i}}=\left(\begin{array}{ccccccccc}
X(\phi^i_{1,2})\\
 & \ddots\\
 &  & X(\phi^i_{1,d})\\
 &  &  &   \ddots\\
 &  &  &  &   X(\phi^i_{2,d})\\
 &  &  &  &  &   \ddots\\
 &  &  &  &  &  &  X( \phi^i_{d-1,d})\\
 &  &  &  &  &  &  &   0_{d-1}
\end{array}\right),
\end{gather}
where
\begin{gather}
X(\phi^i_{k,l})=\left(\begin{array}{cc}
0& \phi^i_{k,l}\\
-\phi^i_{k,l} & 0
\end{array}\right),\, \mr{where},\, \phi^i_{k,l}=\phi^i_k-\phi^i_l,
\end{gather}
and $0_{d-1}$ is $(d-1)\times (d-1)$ zero matrix. Note that $\phi^i_{k,l}\in(-2\pi,2\pi)$. Comparing structures of matrices $\mathrm{Ad}_{D_i}$ and $\mathrm{ad}_{\tilde{D}_{i}}$ we deduce that if all $\phi^i_{i,j}\neq\pm\pi$ then $\mc{C}(\mr{Ad}_\mc{S})=\mc{C}(\mr{ad}_\mc{X})$. The situation is different when $\phi^i_{k,l}=\pm\pi$. In this case  $\mathrm{Ad}_{D_i}$ has additional degeneracies compared to $\mathrm{ad}_{\tilde{D}_{i}}$ as $O(\phi^i_{k,l})=O(\pm\pi)=-I_2$. Let $P$ be the rotation plane corresponding to the angle $\phi^i_{k,l}=\pm\pi$. One can then construct a rotation $O^\prime\in SO(d^2-1)$ whose elementary rotation planes are exactly as in $\mathrm{ad}_{\tilde{D}_{i}}$ except $P$ which is replaced by a plane $P^\prime$, $P\perp P^\prime$. This can be achieved using available $d-1$ directions corresponding to $I_{d-1}$. If the rotation angle along $P^\prime$ is also $\pi$ then $[\mr{Ad}_{U_i},O^\prime]=0$ and $[\mr{ad}_{X_i},O^\prime]\neq 0$. Hence the space $\mc{C}(\mr{Ad}_{U_i})$ is larger than $\mc{C}(\mr{ad}_{X_i})$ and there is possibility that it might be true also for sets  $\mc{C}(\mr{Ad}_\mc{S})$ and $\mc{C}(\mr{ad}_\mc{X})$. As a conclusion we get
\begin{fact}\label{diff-su}
Let $S=\{U_1,\ldots,U_n\}\subset SU(d)$ and $\mc{X}=\{X_1,\ldots, X_n\}$ be the corresponding set of Lie algebra elements (constructed as described in Section \ref{log-construct}). The space $\mc{C}(\mr{Ad}_\mc{S})$ can be larger than $\mc{C}(\mr{ad}_\mc{X})$ if and only if the difference between spectral angles for at least one of the matrices $U_i\in\mc{S}$ is equal to $\pm \pi$.
\end{fact}

\subsubsection{The case of $SO(d)$}\label{sec:adAd__sod}

We consider $\mc{S}=\{O_1,\ldots,O_n\}\subset SO(d)$ and  $\mc{X}=\{X_1,\ldots,X_n\}$ be the corresponding Lie algebra elements (constructed as described in Section \ref{log-construct}). We have $\mc{C}(\mr{ad}_\mc{X})\subseteq \mc{C}(\mr{Ad}_\mc{S})$ and our goal is to characterise the cases when the space $\mc{C}(\mr{Ad}_\mc{S})$ can be strictly larger than $\mc{C}(\mr{ad}_\mc{X})$.  Matrices $O_i$ can be put into a standard form $O_i=V_iR_iV_i^\dagger$, where $V_i\in SO(d)$ and $R_i$ is a block diagonal matrix consisting of $k\leq \floor{\frac{d}{2}}$ two dimensional blocks  representing rotations by angles $\{\phi^i_1,\ldots,\phi^i_{k}\}$, $\phi^i_{j}\in (0,2\pi)$ and one $(d-2k)$-dimensional block that is the identity matrix. Note next that $\mr{Ad}_{O_i}=\mr{Ad}_{V_iR_iV_i^\dagger}=\mr{Ad}_{V_i}\mr{Ad}_{R_i}\mr{Ad}_{V_i}^t$. Each matrix $\mr{Ad}_{R_i}$ can be brought to the standard block diagonal form containing the following blocks 
\begin{enumerate}
\item $O(\phi^i_{a,b})$ and $O(\psi^i_{a,b})$, where $\phi^i_{a,b}=\phi^i_{a}-\phi^i_{b}$, $\psi^i_{a,b}=\phi^i_{a}+\phi^i_{b}$, $a<b$. The number of these blocks is $k(k-1)$.
\item The identity block of dimension $k+\frac{(d-2k)(d-2k-1)}{2}$.
\item Blocks $O(\phi^i_j)$, where $j\in\{1,\ldots,k\}$. Each block $O(\phi^i_j)$ appears $(d-2k)$ times. Hence we have $k(d-2k)$ blocks like this.
\end{enumerate}
Matrices $\mr{ad}_{X_i}$ have the same structure as matrices $\mr{Ad}_{O_i}$ albeit the identity block is replaced by the $0$-block of the same dimension and the rotational blocks $O(\phi^i_{a,b})$, $O(\psi^i_{a,b})$ and $O(\phi^i_j)$ are replaced by $\left(\begin{array}{cc}0&\phi^i_j\\-\phi^i_j&0\end{array}\right)\in\mk{so}(2)$, where every $\phi^i_j\in(0,2\pi)$. Repeating the reasoning for $SU(d)$ we get:

\begin{fact}\label{diff-so}
Let $S=\{U_1,\ldots,U_n\}\subset SO(d)$ and $\mc{X}=\{X_1,\ldots, X_n\}\subset \mathfrak{so}(d)$ be the corresponding set of Lie algebra elements (constructed as described in Section \ref{log-construct}). The space $\mc{C}(\mr{Ad}_\mc{S})$ can be bigger than $\mc{C}(\mr{ad}_\mc{X})$ if and only if the difference or the sum of spectral angles $\phi^i_a$ and $\phi^i_b$ for at least one of the matrices $O_i\in\mc{S}$ is an odd multiple of $\pi$.
\end{fact}
\subsection{Pairs generating infinite subgroups of $G$}\label{sec:pair_generating_inf_subgroups}
In this section we show that elements that are close enough to $Z(G)$ generate $G$ if the corresponding Lie algebra elements generate $\mathfrak{g}$ (see Theorem \ref{thm-alg}). We begin with recalling the elementary properties of the matrix exponential and the matrix logarithm. To this end we define the norm of $A\in\mr{Mat}_d(\mathbb{C})$ by $\|A\|=\sqrt{\mathrm{tr}(AA^{\dagger})}$.

 Next we  recall that the group commutator of two invertible matrices (with respect to matrix multiplication) is defined as $[A,B]_{\bullet }=ABA^{-1}B^{-1}$. Naturally, if matrices commute in a usual sense then $[A,B]_\bullet =I$. The following lemma relates the distance between $[A,B]_\bullet$ and $I$ with the distances of $A$ and $B$ from the identity.
\begin{lem}
\label{lem:infinite}Let $A,B\in G$ where $G=SU(d)$ or $G=SO(d)$ and let $C=[A,B]_{\bullet }$.
We have the following: 
\begin{gather}
\|C-I\|\leq\sqrt{2}\|A-I\|\|B-I\|,\label{ineq1}\\
\mathrm{If}\,[A,C]_{\bullet }=I\,\mathrm{and}\,\|B-I\|<2,\,\mathrm{then}\,[A,B]_{\bullet}=I.
\end{gather}
\end{lem}
\begin{proof}
Can be found in Lemmas 36.15 and 36.16 of \cite{curtis}.
\end{proof}
We next define open balls in $G=SO(d)$ or $SU(d)$ centred around elements from $Z(G)$
and of radius $1/\sqrt{2}$, $B_{\alpha}=\{g\in G:\|g-\alpha I\|<1/\sqrt{2}\}$. Let $\mathcal{B}=\bigcup_{\alpha I\in Z(G)}B_{\alpha }$.
\begin{lem}\label{lem:inf}
Let $g,h\in B_{1}$ and assume $[g,h]_\bullet\neq I$. The group $<g,h>$
generated by $g,h$ is infinite.
\end{lem}
\begin{proof}
Define the sequence $g_{0}=g$, $g_{1}=[g_{0},h]_{\bullet }$, $g_{n}=[g_{n-1},h]_{\bullet }$. By our assumptions $\|h-I\|=d\leq1/\sqrt{2}$ . Therefore using Lemma \ref{lem:infinite}
\begin{gather*}
\|g_{n}-I\|\leq\sqrt{2}d\|g_{n-1}-I\|.
\end{gather*}
Thus $\|g_{n}-I\|\leq(\sqrt{2}d)^{n}\|g-I\|$ and $g_{n}\rightarrow I$, when $n\rightarrow \infty$.
Assume that the sequence is finite, i.e. for some $N$ we have $g_{N}=I$.
That means $[g_{N-1},h]_{\bullet }=I$. But $g_{N-1}=[g_{N-2},h]_{\bullet }$ and clearly
$\|g_{k}-I\|<2$ and by Lemma \ref{lem:infinite}, $[g_{N-2},h]_{\bullet }=I$. Repeating this argument
we get $[g,h]_{\bullet }=I$ which is a contradiction. Therefore $<g,h>$ is infinite. 
\end{proof}
\begin{coll}\label{coll:balls}
Let $g\in B_{\alpha_1}$ and $h\in B_{\alpha_2}$, where $\alpha_1$ and $\alpha_2$ are such that $\alpha_1I,\alpha_2 I\in Z(G)$ and assume $[g,h]_\bullet\notin Z(G)$. Then the  group $<g,h>$ is infinite.\label{coll:infinite_with_exponents}
\end{coll}
\begin{proof}
If $\alpha_1=\alpha_2=1$ the result follows from Lemma \ref{lem:inf}. For all other $\alpha_i$'s let $g^\prime=\alpha_1 ^{-1}g$ and $h^\prime=\alpha_2^{-1}h$. Then $h^\prime,g^\prime\in B_1$ and  $[g^\prime,h^\prime]_{\bullet }\neq I$. Thus by Lemma \ref{lem:inf}, $<g^\prime,h^\prime>$ is infinite. Note that $<g,h>$ is up to the finite covering equal to $<g^\prime,h^\prime> $ and therefore is infinite too. 
\end{proof}
\begin{figure}[ht!]
\centering\includegraphics*[scale=0.4]{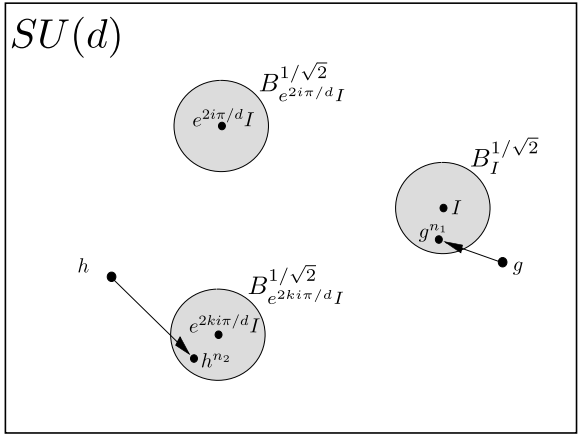}
\caption{The group $SU(d)$ with the exemplary open balls $B_\alpha$ centred at elements form $Z(SU(d))$.}
\end{figure}

We next provide explicit conditions for elements of $G$ to belong to balls $ B_{\alpha }$. To this end let $\alpha_m I$ be the elements of $Z(G)$. We have the following
\begin{gather}
\|g-\alpha_mI\|^2=\mathrm{tr}(g-\alpha_mI)(g^\ast-\alpha^\ast_mI)=2\mathrm{tr}I-\alpha^\ast_m\mathrm{tr}g-\alpha_m\mathrm{tr}g^\ast.
\label{def:distance_with_alpha}
\end{gather}
For $SU(d)$ we have $\alpha_m^d=1$ and hence $\alpha_m=\cos\theta_m+i\sin\theta_m$, where $\theta_m=\frac{2m\pi}{d}$ and $m\in\{1,\ldots,d\}$. Let $\{e^{i\phi_1},e^{i\phi_2},\ldots,e^{i\phi_d}\}$ be the spectrum of $U_d\in SU(d)$. The conditions for $U_d\in SU(d)$ to belong to the ball $B_{\alpha_m}$ read:

\begin{gather}\label{condB1}
U_d\in B_{\alpha_m} \Leftrightarrow\sum_{i=1}^{d}\sin^{2}\frac{\phi_{i}-\theta_m}{2}<\frac{1}{8},\,\,\,\,\sum_{i=1}^{d}\phi_{i}=0\,\mathrm{mod}\,2\pi.
\end{gather}
For $SO(2k+1)$ the centre is trivial and we have only one ball $B_1$. Let $\{1,e^{i\phi_1},e^{-i\phi_1},\ldots,e^{i\phi_k},e^{-i\phi_k}\}$ be the spectrum of  $O_{2k+1}\in SO(2k+1)$. We have 
\begin{gather}\label{condB2}
O_{2k+1}\in B_1\Leftrightarrow \sum_{i=1}^{k}\sin^{2}\frac{\phi_{i}}{2}<\frac{1}{16}.
\end{gather}
Finally $Z(SO(2k))=\{I,-I\}$ and we have two balls $B_1$, $B_{-1}$. Let $$\{e^{i\phi_1},e^{-i\phi_1},\ldots,e^{i\phi_k},e^{-i\phi_k}\},$$ be the spectrum of $O_{2k}$. The conditions for the spectral angles are as follows 
\begin{gather}\label{condB3}
O_{2k}\in B_{1} \Leftrightarrow \sum_{i=1}^{k}\sin^{2}\frac{\phi_{i}}{2}<\frac{1}{16},\\\label{condB4}
O_{2k}\in B_{-1} \Leftrightarrow \sum_{i=1}^{k}\sin^{2}\frac{\phi_{i}-\pi}{2}<\frac{1}{16}.
\end{gather}
\begin{thm}\label{thm-alg}
Let  $G=SO(d)$ or $G=SU(d)$. Let $\mathcal{S}=\{g_1,\ldots,g_n\}\subset G$ be such that $g_i\in B_{\alpha}$, where $\alpha I\in Z(G)$ and let $\mc{X}=\{X_1,\ldots,X_n\}\subset \mk{g}$ be the Lie algebra elements assigned to $\mc{S}$ (constructed as described in Section \ref{log-construct}). $\mc{S}$ generates $G$ if and only if $\mc{X}$ generates $\mk{g}$.
\end{thm}
\begin{proof}
By Lemma \ref{Ad2}, matrices $\mc{S}$ generate $G$ if they generate an infinite subgroup and $\mc{C}(\mr{Ad}_{\mc{S}})=\mc{C}(\mr{Ad}_{G})$. The cases when spaces $\mc{C}(\mr{Ad}_{\mc{S}})$ and  $\mc{C}(\mr{ad}_{\mc{X}})$ can differ are characterised by Facts \ref{diff-su} and \ref{diff-so}. Assume that $\mc{S}\subset SU(d)$. The spaces $\mc{C}(\mr{Ad}_{\mc{S}})$ and  $\mc{C}(\mr{ad}_{\mc{X}})$ can differ if and only if for one of the matrices $g_i\in\mc{S}$ we have $\phi^i_{a,b}=k\pi$, where $k$ is odd. But then $\phi^i_a=\phi^i_b\pm\pi$ and for some $\theta_m=\frac{2\pi m}{d}$
\[
\sin^{2}\frac{\phi^i_{b}\pm\pi-\theta_m}{2}+\sin^{2}\frac{\phi^i_{b}-\theta_m}{2} =1,
\]
which means $g_i$ does not satisfy (\ref{condB1}). Assume next that $\mc{S}\subset SO(d)$. The spaces $\mc{C}(\mr{Ad}_{\mc{S}})$ and  $\mc{C}(\mr{ad}_{\mc{X}})$ can differ iff the difference or the sum of spectral angles $\phi^i_a$ and $\phi^i_b$ is equal to an odd multiple of $\pi$. For odd $d$ we arrive at
\[
\sin^{2}\frac{\pm\phi^i_{b}\pm\pi}{2}+\sin^{2}\frac{\phi^i_{b}}{2} =1,
\]
and for even $d$ we additionally have
\[
\sin^{2}\frac{\pm\phi^i_{b}\pm\pi-\pi}{2}+\sin^{2}\frac{\phi^i_{b}-\pi}{2} =1,
\]
which means $g_i$ does not satisfy (\ref{condB2}), (\ref{condB3}) or (\ref{condB4}). 
\end{proof}
\subsection{Universal sets for  $G$}
In this section we consider situation when not all the matrices belonging to $\mc{S}$  are contained in $\mc{B}$. We already know that if there are two elements $g,h\in \overline{<S>}\cap\mc{B}$ such that $[g,h]_\bullet\notin Z(G)$ than the group $\overline{<S>}$ is infinite. It turns out that for $\mc{S}$ that satisfies the necessary universality condition, i.e. $\mc{C}(\mr{Ad}_{g_1},\ldots,\mr{Ad}_{g_k})=\{\lambda I\}$ this is actually an equivalence relation. 

\begin{lem}\label{lem19}
Let $\mc{S}=\{g_1,\ldots,g_k\}\subset G$ be such that $\mc{C}(\mr{Ad}_{g_1},\ldots,\mr{Ad}_{g_k})=\{\lambda I\}$. The group $\overline{<S>}$ is infinite \textbf{if and only if} there are at least two elements $g,h\in \overline{<S>}\cap\mathcal{B}$ satisfying $[g,h]_\bullet\notin Z(G)$.
\end{lem}
\begin{proof}
Assume $\overline{<S>}$ is infinite. Then under the assumption $\mc{C}(\mr{Ad}_{\mc{S}})=\{\lambda I\}$ we have 
$\overline{<S>}=G$. Thus balls  $B_{\alpha}$ must contain elements of $<S>$ commuting to a noncentral elements and the result follows. On the other hand if  there are at least two elements $g,h\in \overline{<S>}$ such that they belong to some balls $B_{\alpha }$, where $\alpha I\in Z(G)$, and $[g,h]_\bullet\notin Z(G)$ then by Corollary \ref{coll:infinite_with_exponents} $\overline{<S>}$ is infinite. 
\end{proof}
We already know that the necessary universality condition places significant constraints on the structure of the infinite $\overline{<\mc{S}>}$. It turns out that this is the case also when  $\overline{<\mc{S}>}$ is finite. The constrains regard the structure of $<\mc{S}>\cap \mc{B}$. 
\begin{lem}\label{lem20}
Let $\mc{S}=\{g_1,\ldots,g_k\}\subset G$ be such that $\mc{C}(\mr{Ad}_{g_1},\ldots,\mr{Ad}_{g_k})=\{\lambda I\}$. Then either the intersection of $<S>$ with $\mc{B}$ is dense in $\mc{B}$ or is a subgroup of  $Z(G)$. In the first  case $\overline{<S>}=G$ and in the second one $\overline{<S>}$ is finite. 
\end{lem}

\begin{proof}
The group $\overline{<S>}$ can be either infinite or finite. When it is infinite, then by the necessary universality condition, i.e.  $\mc{C}(\mr{Ad}_{g_1},\ldots,\mr{Ad}_{g_k})=\{\lambda I\}$, we have $\overline{<S>}=G$ and it is obvious that $\mc{B}\cap\mc{<S>}$ is dense in $\mc{B}$. Assume next that $<S>$ is finite. By Corollary \ref{coll:infinite_with_exponents} the group commutators of elements from  $\mc{B}\cap\mc{<S>}$ belong to $Z(G)$. We first show that in fact they are equal to the identity, i.e. elements from $\mc{B}\cap\mc{<S>}$ commute.  To see this let $h_1\in B_{\alpha_1}$ and $h_2\in B_{\alpha_2}$. Assume $[h_1,h_2]_\bullet\in Z(G)$. One can always find $\tilde{h}_1,\tilde{h}_2\in B_1$ such that $h_1=\alpha_1 \tilde{h}_1$ and $h_2=\alpha_2\tilde{h}_2$. We have:
\begin{gather}
[h_1,h_2]_\bullet=[\alpha_1 \tilde{h}_1,\alpha_2\tilde{h}_2]_\bullet=\alpha_1 \tilde{h}_1\alpha_2\tilde{h}_2\alpha_1^{-1} \tilde{h}_1^{-1}\alpha_2^{-1}\tilde{h}_2=[\tilde{h}_1,\tilde{h}_2]_\bullet.
\end{gather}
But by inequality (\ref{ineq1}) we have $[\tilde{h}_1,\tilde{h}_2]_\bullet\in B_1$ and it is also easy to see that $B_{\alpha_i}$'s are disjoint. Thus $[h_1,h_2]_\bullet=I$. Next we note that each $B_{\alpha }\cap\mc{<S>}$ is invariant under the conjugation by elements form $G$. Let $\{h_1,\ldots,h_m\}$ be all elements from $B_\alpha\cap <S>$. Once again we can find elements $\{\tilde{h}_1,\ldots, \tilde{h}_m\}\subset B_1$ satisfying $h_i=\alpha\tilde{h}_i$. Let $\mathfrak{g}\ni X_i=\log{\tilde{h}_i}$ (constructed as described in Section \ref{log-construct}). Thus elements of $B_\alpha\cap <S>$ are of the form $\{\alpha e^{X_1},\ldots, \alpha e^{X_m}\}$. We also know that $B_\alpha\cap <S>$ is $\mr{Ad}_{\mc{S}}$ invariant, i.e. 
\begin{gather}
g_i\alpha e^{X_j}g_i^{-1}=\alpha \mr{Ad}_{g_i}e^{X_j}=\alpha e^{X_r},\, g_i\in\mc{S},
\end{gather}
where $i\in\{1,\ldots,k\}$ and $j,r\in\{1,\ldots,m\}$. Thus we have $\mr{Ad}_{g_i}e^{X_j}=e^{X_r}$.  As the distance from the identity of the left and right side is smaller than $1$ we have $\log{\mr{Ad}_{g_i}e^{X_j}}=\log{e^{X_r}}$. By the construction, $\log{e^{X_r}}=X_r$ and from our  definition of logarithm: $\log{\mr{Ad}_{g_i}e^{X_j}}=\mr{Ad}_{g_i} \log e^{X_j}=\mr{Ad}_{g_i}X_j$. Hence $\mr{Ad}_{g_i}X_j=X_r$  and the subspace $\{X_1,\ldots,X_m\}\subset \mathfrak{g}$ is an invariant subspace for all matrices $\{\mathrm{Ad}_{g_1},\ldots\mathrm{Ad}_{g_k}\}$. By the condition $\mc{C}(\mr{Ad}_{g_1},\ldots,\mr{Ad}_{g_k})=\{\lambda I\}$ this subspace must be either $0$ or $\mathfrak{g}$. Assume it is $\mk{g}$. Recall that we have:
\[
[\alpha e^{X_j},\alpha e^{X_j}]=0,\,\, i,j\in\{1,\ldots,k\}.
\]
Thus there is $U$ such that $\alpha e^{X_i}=\alpha e^{UD_iU^{-1}}$, where $D_i$ is diagonal. Hence $X_i=UD_iU^{-1}$. Thus matrices $\{X_1,\ldots, X_m\}$ commute and we get a contradiction. Hence $<S>\cap B_{ \alpha}$ is either empty or $\alpha I$.  The result follows.
\end{proof}
Lemma \ref{lem20} leads to the following conclusion:
\begin{coll}\label{finite}
Let $\mc{S}=\{g_1,\ldots,g_k\}\subset G$ be such that $\mc{C}(\mr{Ad}_{g_1},\ldots,\mr{Ad}_{g_k})=\{\lambda I\}$. Then $<\mc{S}>$ is infinite \textbf{if and only if} there is an element in $<\mc{S}>$ that belongs to $\mc{B}$ and does not belong to $Z(G)$.
\end{coll} 
Of course $\mc{S}$ can be such that its elements do not belong to $\mc{B}$. In the following we show that by taking powers we can move every element of $G$ into $B_{\alpha}$ for some $\alpha I\in Z(G)$. Moreover there is a global upper bound for the required power. 
\begin{fact}\label{fact:max_exponent}
For groups $G=SU(d)$ and $G=SO(d)$ there is $N_G\in\mathbb{N}$ such that for
every $g\in G$, $g^{n}\in B_{\alpha_m}$ for some $\alpha_mI\in Z(G)$ and $1\leq n\leq N_G$.
\end{fact}
\begin{proof}
Let us first recall that by the Dirichlet theorem (see Theorem 201 in \cite{hardy}),  for given real numbers $x_1,\,x_2,\ldots,x_k$ we can find $n\in\mathbb{N}$ so that $nx_1,\ldots,nx_k$ all differ from integers by as little as we want. Let $\{\phi_1,\ldots,\phi_k\}$ be the spectral angles of $g\in G$ and let  $\phi_i=2\pi x_i$, where $x_i\in [0,1)$. By Dirichlet theorem we can always find $n$ such that $nx_i$'s are close enough to integers to make $g^n$ to belong to $B_1$. For  $g\in G$ let $n_g$ be the smallest positive integer such that $g^{n_g}\in B_{\alpha }$ for some $\alpha I\in  Z(G)$ (by Dirichlet theorem we know that $n_g< \infty$). Let $\mathcal{O}^{n_g}_g$ be an open neighbourhood\footnote{This kind of a neighbourhood exists as taking powers is a continuous operation.} of $g$ such that for any $h\in\mathcal{O}^{n_g}_g$ we have $h^{n_g}\in B_ {\alpha }$. Note that there might be some $h\in \mathcal{O}^{n_g}_g$ for which $n_g$ is not optimal but this will not play any role. Let $\{\mathcal{O}^{n_g}_g\}_{g\in G}$ be the resulting open cover of $G$. As $G$ is compact there is a finite subcover $\{\mathcal{O}^{n_{g_i}}_{g_i}\}$ and hence $N_G=\mathrm{sup}_in_{g_i}$ is well defined and finite.   
\end{proof}

For $g\in G$ let $1\leq n_g\leq N_G$ denote the smallest integer such that $g^{n_g}\in \mathcal{B}$. Using Corollary \ref{finite} we deduce that $<\mc{S}>$ is finite if and only if for every $g\in <S>$ we have $g^{n_g}\in Z(G)$. This in turn places certain constrains on the spectra of elements belonging to $<\mc{S}>$. 
\begin{definition}\label{exangle}
Assume $g\notin \mc{B}$.  The spectrum of $g$ is exceptional if for some $1\leq n\leq N_G$ we have $g^n\in Z(G)$.
\end{definition}

In other words the spectrum of $g$ is exceptional iff (1) $g\in SU(d)$ and all spectral elements of $g$ are $n^\mr{th}$ roots of $\alpha\in \mathbb{C}$, where $\alpha^d=1$, for some fixed $1\leq n\leq N_{SU(d)}$, (2) $g\in SO(2k+1)$ and all spectral elements of $g$ are $n^\mr{th}$ roots of unity for some fixed $1\leq n\leq N_{SO(2k+1)}$, (3) $g\in SO(2k)$ and all spectral elements of $g$ are $n^\mr{th}$ roots of $\alpha$, where $\alpha^2=1$, for some fixed $1\leq n\leq N_{SO(2k+1)}$.
Note that the set of exceptional spectra is a finite set. As a direct consequence we get the following result:
\begin{thm}\label{main}
Let $\mc{S}=\{g_1,g_2,\ldots,g_k\}\subset G$, where $G=SO(d)$ and $d\neq 4$ or $G=SU(d)$. Assume $\mc{C}(\mr{Ad}_{g_1},\ldots,\mr{Ad}_{g_k})=\{\lambda I\}$ and that there is at least one element in $\mc{S}$ for which the spectrum is not exceptional. Then $\overline{<\mc{S}>}=G$.
\end{thm}

\subsection{The algorithm for checking universality}\label{alg}
In this section we present a simple algorithm that allows to decide universality of any given set of gates $\mc{S}\subset G$ in a finite number of steps. It works for $G=SU(d)$ and $G=SO(k)$, besides $k=4$.

\textbf{The Algorithm for checking universality of $\mc{S}=\{g_1,\ldots,g_n\}$}
\begin{description}
\item[Step 1] Check if $\mathcal{C}(\mathrm{Ad}_\mathcal{S})=\{\lambda I\}$. If the answer is NO stop as the set $\mathcal{S}$ is not universal. If YES, set $l=1$ and go to step 2.
\item[Step 2] Check if there is a matrix $g\in \mc{S}$ for which $g^{n_g}$ belongs to $\mc{B}$ but not to $Z(G)$, where $1\leq n_g\leq N_G$. If so $\mathcal{S}$ is universal. If NO,  set $l=l+1$.
\item[Step 3] Define the new set $\mc{S}$ by adding to $\mathcal{S}$ words of length $l$, i.e products of elements from $\mathcal{S}$ of length $l$. If the new $\mc{S}$ is equal to the old, the group  $<\mc{S}>$ is finite. Otherwise go to step 2.  
\end{description}
If the group generated by $\mc{S}$ is finite the algorithm terminates in step 3 for some $l<\infty$. Otherwise it terminates in step 2 for $l<\infty$. In the following we discuss the bounds for $l$. In case when the group generated by $\mc{S}$ is finite the upper bound for $l$ is the order of largest finite subgroup of $SU(d)$. When the set $\mc{S}$ is symmetric, i.e. $\mc{S}=\{U_1,\ldots, U_k,U_1^{-1},\ldots,U_k^{-1}\}$ and the group generated by $\mc{S}$ is infinite the bound for $l$ can be determined by looking at the averaging operator $T_\mc{S}:L^2(SU(d))\rightarrow L^2(SU(d))$:
\begin{gather}
\left(T_\mc{S}f\right)(g)=\frac{1}{2k}\sum_{i=1}^k\left(f(U_ig)+f(U_i^{-1}g\right).
\end{gather}
Let $\|T\|_\mr{op}:=\mr{sup}_{f\in L^2(SU(d))}\frac{\|Tf\|_2}{\|f\|_2}$, where $\|\cdot\|_2$ is the usual $L^2$ norm. One easily checks that shifting operators $(\tilde{U}f)(g)=f(U^{-1}g)$ are unitary and hence their operator norm is $1$. Thus, using triangle inequality, we see that $\|T_\mc{S}\|_{\mr{op}}\leq 1$. In fact the constant function $f=1$ is the eigenvector of $T_\mc{S}$ with the eigenvalue $1$ and $\|T_\mc{S}\|_{\mr{op}}=1$. Let $L_0^2(SU(d))$ be the subspace of $L^2(SU(d))$ containing functions with the vanishing mean. Consider operator $T_\mc{S}|_{L_0^2(SU(d))}$. The norm of this operator is $1$ if and only if $1$ is an accumulation point of the spectrum of $T_\mc{S}$. Otherwise it is strictly less than $1$ and we will denote it by $\lambda_1$.  If this is the case we say that $T_\mc{S}$ has a spectral gap. The recent results \cite{BG1, BG2} ensure that $T_\mc{S}$ has a gap at least when matrices from $\mc{S}$ have algebraic entries. For transcendental entries the problem of the spectral gap existence is open. In fact, Sarnak conjectures the spectral gap is present for any universal set. The existence of spectral gap has interesting implications. As was shown in \cite{harrow} (our formulas are slightly different than in \cite{harrow} as we use Hilbert-Schmidt norm):
\begin{fact} \label{Harrow}
Let $\mathcal{S}$ be an universal, symmetric set of gates and assume $T_\mc{S}$ has a spectral gap. Let $\lambda_1=\|T_\mc{S}|_{L_0^2(SU(d))}\|_{\mr{op}}$. For every $U\in SU(d)$, $\epsilon>0$ and 
$$n>A\log\left(\frac{1}{\epsilon}\right)+B$$
 there is $U_n\in W_n(\mathcal{S})$ such that $\|U-U_n\|<\epsilon$, where
\begin{gather*}
A=\frac{d^2-1}{\log\left(1/\lambda_1\right)},\,\,\,\,\,\,B=\frac{\log\left(2^{d^2-1}/a_1\right)+\frac{1}{2}(d^2-1)\log(d^2-1)}{\log\left(1/\lambda_1\right)}
\end{gather*}
and $a_1$ is such that for any ball of radius $\epsilon$ in $SU(d)$ its volume (with respect to normalised Haar measure) $V_{B_\epsilon}$, satisfies
\[
 V(B_\epsilon)\geq a_1\epsilon^{d^2-1}.
\]
\end{fact}

The upper bound for $l$ in our algorithm in case when $<\mc{S}>$ is infinite  is given by the minimal number of gates that are needed to approximate an element whose distance from $\mc{B}$ is equal $\frac{1}{2\sqrt{2}}$ with the precision $\epsilon=\frac{1}{2\sqrt{2}+\delta}$, where $\delta$ is arbitrarily small positive number. Using Fact \ref{Harrow} this number is bounded by:
\begin{gather}\label{ub}
l\leq\frac{d^2-1}{\log\left(1/\lambda_1\right)}\log(2\sqrt{2}+\delta)+\frac{\log\left(2^{d^2-1}/a_1\right)+\frac{1}{2}(d^2-1)\log(d^2-1)}{\log\left(1/\lambda_1\right)}
\end{gather}
\noindent Moreover, by explicit calculation the volume 
\[
V(B_\epsilon)=\frac{1}{\pi}\left(2\arcsin\frac{\epsilon}{2\sqrt{2}}-\frac{1}{2}\sin 4\arcsin\frac{\epsilon}{2\sqrt{2}} \right).
\]
One easily checks that $a_1\epsilon^3$, where 
\[
a_1=\frac{16\sqrt{2}}{\pi}\left(2\arcsin\frac{1}{8}-\frac{1}{2}\sin \left(4\arcsin\frac{1}{8}\right ) \right),
\]
satisfies $V(B_\epsilon)\geq a_1\epsilon^3$ for $\epsilon\in[0,\frac{1}{2\sqrt{2}}]$.

Finally, we note that when spectral gap is small, i.e. $\lambda_1$ is close to $1$ the upper bound given by \ref{ub} can be in fact very big. This is the case, for example, when matrices $\mc{S}$ are very close to some matrix $U\in SU(d)$. But then they can be simultaneously introduced to a ball $B_\alpha$ and deciding their universality requires  actually $l=1$. Thus it seems that the bound given in \ref{ub} is useful only if $\lambda_1$ is well separated from $1$. 

\section{Computing $N_G$}

In this section we find upper bounds for $N_{SU(d)}$ and $N_{SO(d)}$ using Dirichlet's approximation theorem \cite{dirichlet,hardy}. These bounds are used in the algorithm presented in Section \ref{alg}.

\begin{thm}\label{d1}
For a given real number $a$ and a positive integer $N$ there exist integers $1\leq n\leq N$ and $p$ such, that $n\phi$ differs from $p$ by at  most $\frac{1}{N+1}$, i.e.
\begin{gather}\label{dirichlet_one}
|na-p|\leq\frac{1}{N+1}.
\end{gather}
\end{thm}

We will use Theorem \ref{d1} in calculation of $N_G$ for $G=SO(3)$ and $G=SU(2)$ - these are two cases when $g\in G$ has a one spectral angle. The  simultaneous version of Dirichlet's theorem  gives a similar approximation for a collection of real numbers  $\phi_1,\ldots,\phi_k$. We will use it for $SO(2k+1)$.

\begin{thm}\label{dn}
For given real numbers $a_1,\ldots,a_d$ and a positive integer $N$ there exist integer $1\leq n\leq N$ and integers $p_1,\ldots,p_k$ such that\begin{gather}\label{dirichlet_many}
|na_i-p_i|\leq\frac{1}{(N+1)^{1/d}}.
\end{gather}
\end{thm}
For groups $SO(2k)$ and $SU(d)$ we need to prove a modified version of Dirichlet's theorem. To this end for any  $x\in\mb{R}$ and  $d\in\mb{Z}_+$ we define $\{x\}_k$ to be the difference between $x$ and the largest $p+\frac{k}{d}$ that is smaller or equal  to $x$, where $p\in \mathbb{Z}$, $k\in\{0,1,\ldots,d-1\}$. Clearly $\{x\}_k\in[0,1)$. For $x=(x_1,\ldots,x_m)\in \mathbb{R}^m$ we define $\{x\}_k=(\{x_1\}_k,\ldots,\{x_m\}_k)$. Let $\mathcal{L}_{m,d}$ be the lattice in $\mathbb{R}^m$ given by points 
\[
(q_1,\ldots,q_m),\,(q_1+\frac{1}{d},\ldots,q_m+\frac{1}{d}),\,\ldots ,(q_1+\frac{d-1}{d},\ldots,q_m+\frac{d-1}{d}),
\] 
where $q_1,\ldots q_m\in \mathbb{Z}$. An important property of the lattice $\mathcal{L}_{m,d}$ is that for any $p,q\in\mathcal{L}_{m,d}$ we have $p\pm q\in\mathcal{L}_{m,d}$. As a direct consequence of this property we get the following theorem.

\begin{thm}\label{mod}
For $a=(a_1,\ldots,a_m)$ and positive $\epsilon<\frac{1}{2d}$ there exist: integer $1\leq n\leq \ceil*{\frac{1}{d\epsilon^m}}$ and a point $p=(p_1,\ldots,p_m)\in \mathcal{L}_{m,d}$ such that $\forall i\in\{1,\ldots,m\}$:
\begin{gather}\label{dirichlet_many}
|na_i-p_i|<\epsilon. 
\end{gather}
\end{thm}
\begin{proof}
For a given point $a=(a_1,\ldots,a_m)\in \mathbb{R}^{m}$ consider $dQ^m+1$ points:
\begin{gather}\label{points}
\{na\}_0, \,\{na\}_1,\ldots,\{na\}_{d-1},\,\, n\in\{0,\ldots,Q^m\}
\end{gather}
Next take an $m$-dimensional cube $[0,1)^m$ and divide it into $dQ^m$ boxes by drawing planes parallel to its faces at distances $\frac{1}{\sqrt[m]{d}Q}$. By Dirichlet's pigeon hole principle, at least two points from (\ref{points}) fall to the same box. Let these points be $\{q_1a\}_i$  and  $\{q_2a\}_j$, where $i,j\in\{1,\ldots,d-1\}$ and $q_1<q_2$.  Note that $q_1$ cannot be equal to $q_2$ as in this case $\epsilon>\frac{1}{2d}$. As the lattice $\mathcal{L}_{m,d}$ is invariant with respect to addition and subtraction of its points we have  $\max_l|\{(q_2-q_1)a_l\}_{k}|<\frac{1}{\sqrt[m]{d}Q}$, where $k=j-i$ if $i<j$ or $k=d+j-i$ when $i>j$. The result follows.
\end{proof}
We begin with finding the exact values of $N_{SU(2)}$ and $N_{SO(3)}$.
\begin{fact}
$N_{SO(3)}=12$ and $N_{SU(2)}=6$. \label{fact:max_exp_su2s3}
\end{fact}
\begin{proof}
Let  $O\in SO(3)$ and let $[0,2\pi)\ni\phi=2a\pi$ be its spectral angle. By Theorem \ref{d1} for a given $N$ there are integers $p$ and $1\leq n \leq N$ such that $|na-p|\leq\frac{1}{N+1}$. Multiplying this inequality by $\pi$ yields $|n\frac{\phi}{2}-p\pi|\leq\frac{\pi}{N+1}$. Note that (\ref{condB2}) simplifies to $|\sin\frac{\psi}{2}|<\frac{1}{4}$, i.e. for a given $\phi$ we look for $n$ such that  $|n\frac{\phi}{2}-p\pi|<\arcsin\frac{1}{4}$. Combining these two observations we need to find the smallest $N$ such that $\frac{\pi}{N+1}<\arcsin\frac{1}{4}$. It is 
\begin{gather}
N=\ceil*{\frac{\pi-\arcsin\frac{1}{4}}{\arcsin\frac{1}{4}}}=12.
\label{nmaxSo31}
\end{gather} 
Formula (\ref{nmaxSo31}) gives an upper bound for $N_{SO(3)}$. Note however that for $\frac{\phi}{2} =\arcsin\frac{1}{4}$ the smallest $n$ such that $|n\arcsin\frac{1}{4}-\pi|<\arcsin\frac{1}{4}$ is exactly $12$ (see figure \ref{wykresy}(a)), hence $N_{SO(3)}=12$.
\begin{figure}[ht!]
\begin{center}\includegraphics[scale=0.4]{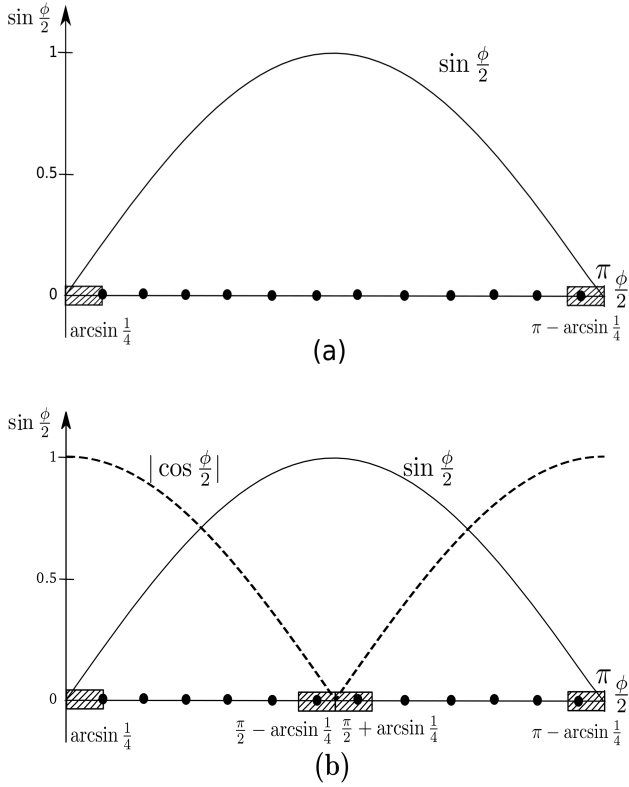}
\end{center}
\caption{\label{wykresy} (a) Condition (\ref{condB2}) for $SO(3)$. Black dots correspond to $n\arcsin\frac{1}{4}$ and dashed segments are determined by $|\sin\frac{\phi}{2}|<\frac{1}{4}$, (b) Conditions (\ref{condB1}) for $U\in SU(2)$. Black dots corresponds to $n\arcsin\frac{1}{4}$ and dashed segments are determined by $|\sin\frac{\phi}{2}|<\frac{1}{4}$ or $|\sin\frac{\phi-\pi}{2}|<\frac{1}{4}$.}
\end{figure}

Assume next $U\in SU(2)$ and $[0,2\pi)\ni\phi=a\pi$  be its spectral angle. By Theorem \ref{d1} for a given $N$ there are integers $p$ and $1\leq n \leq N$ such that  $|na-p|\leq\frac{1}{N+1}$. Multiplying this inequality by $\frac{\pi}{2}$ yields $|n\frac{\phi}{2}-p\frac{\pi}{2}|\leq\frac{\pi}{2(N+1)}$. Note that (\ref{condB2}) simplifies to $|\sin\frac{\psi}{2}|<\frac{1}{4}$ or $|\sin\frac{\psi-\pi}{2}|<\frac{1}{4}$, i.e. for a given $\phi$ we look for $n$ such that  $|n\frac{\phi}{2}-p\frac{\pi}{2}|<\arcsin\frac{1}{4}$. Combining these two observations we need to find the smallest $N$ such that $\frac{\pi}{2(N+1)}<\arcsin\frac{1}{4}$. This is 
\begin{gather}
N=\ceil*{\frac{\frac{\pi}{2}-\arcsin\frac{1}{4}}{\arcsin\frac{1}{4}}}=6.
\label{nmaxSu3}
\end{gather} 
Formula (\ref{nmaxSu3}) gives an upper bound for $N_{SU(2)}$. Note however that for $\frac{\phi}{2} =\arcsin\frac{1}{4}$ the smallest $n$ such that $|n\arcsin\frac{1}{4}-\frac{\pi}{2}|<\arcsin\frac{1}{4}$ is exactly $6$ (see figure \ref{wykresy}(b)). Hence $N_{SU(2)}=6$.
\end{proof}

\begin{fact}\label{sodN}
The values of $N_{SO(2k+1)}$ and $N_{SO(2k)}$ are bounded from the above by:
\begin{gather}
\label{nmax_so2kp1}
N_{SO(2k+1)}<\ceil*{\left(\frac{\pi}{\arcsin\frac{1}{4\sqrt{k}}}\right)^{k}},\\\label{nmax_so2k}
N_{SO(2k)}<  \ceil*{\frac{1}{2}\left(\frac{\pi}{\arcsin\frac{1}{4\sqrt{k}}}\right)^{k}}.
\end{gather}
\end{fact}
\begin{proof}

The spectral angles of $O\in SO(d)$ are $\{\phi_1,-\phi_1,\ldots,\phi_k,-\phi_k \}$ if $d=2k$ or $\{\phi_1,-\phi_1,\ldots,\phi_k,-\phi_k,0 \}$ if $d=2k+1$. We first address the case of $SO(2k)$. Assume that $\phi_i=a_i\pi$ for all $i\in\{1,\ldots,k\}$. The lattice $\pi\cdot\mc{L}_{k,2}$ corresponds exactly to points $\{\frac{\phi_1}{2},\ldots, \frac{\phi_k}{2}\}$ at which balls $B_1$ and $B_{-1}$ given by conditions (\ref{condB3}) and (\ref{condB4}) are centred. Let us next find the smallest hypercube $[-\frac{\beta_k}{2},\frac{\beta_k}{2}]^{\times k}$ contained in the ball $B_1$. By symmetry, its edge length will be the same for $B_{-1}$. To this end one needs to minimise $\sum_i\phi_i^2$ under the condition $\sum_i\sin^2\phi_i=\frac{1}{16}$. Calculations with the use of the Lagrange multipliers show that the coordinates of the minimizing point are all equal and hence  $k\sin^2\frac{\beta_k}{2}=\arcsin\frac{1}{16}$. That means $ \frac{\beta_k}{2}=\arcsin\frac{1}{4\sqrt{k}}$ is the half of the edge length of the largest hypercube contained in a ball $B_{\pm 1}$. We next apply Theorem \ref{mod} to the lattice  $\mc{L}_{k,2}$ and the point $a=(a_1,\ldots,a_k)$ with  $\epsilon=\frac{\arcsin\frac{1}{4\sqrt{k}}}{\pi}< \frac{1}{4}$. As a result we obtain point $p\in \mc{L}_{k,2}$ such that:
\begin{gather}\label{dirichlet_many}
|na_i-p_i|<\frac{\arcsin\frac{1}{4\sqrt{k}}}{\pi},
\end{gather}
where 
\[
n<\ceil*{\frac{\pi^k}{2(\arcsin\frac{1}{4\sqrt{k}})^k}}.
\]
For $SO(2k+1)$ we can directly apply Theorem \ref{dn}. Looking at the hypercube that is contained in one of the balls given by conditions (\ref{condB3}) and (\ref{condB4}) we get the desired result.  
\end{proof}

\begin{fact}\label{sudN}
For $d\geq 3$ the value of $N_{SU(d)}$ is bounded from the above by:
\[
N_{SU(d)}< \ceil*{\frac{1}{d}\left(\frac{2\pi}{\beta_d}\right)^{d-1}}, 
\]
where $\beta_d$ is such that $(d-1)\sin^2\frac{\beta_d}{2}+\sin^2\frac{(d-1)\beta_d}{2}=\frac{1}{8}$.
\end{fact}
\begin{proof}
For $U\in SU(d)$ let $\{\phi_1,\ldots, \phi_d\}$  be the spectral angles of $U$. Assume that for every $i\in\{1\,\ldots,d-1\}$ we have $[0,2\pi)\ni\phi_i=a_i\pi$. As $\sum_{i}\phi_i=0\,\mr{mod}\,2\pi$ we can always put $\phi_d=-\sum_{i=1}^{d-1}{\phi_i}$. We need to first find the edge length of the largest hypercube $[-\frac{\beta_d}{2},\frac{\beta_d}{2}]^{\times(d-1)}$ contained in the ball $B_{1}$. By  symmetry of condition (\ref{condB2}), this length will be the same for other balls. We need to minimise $\sum_i\phi_i^2$ under the condition $\sum_{i=1}^{d-1}\sin^2\phi_i+\sin^2(\sum_{i=1}^{d-1}\phi_i)=\frac{1}{8}$. Calculations with the use of the Lagrange multipliers show that the coordinates of the minimizing point are all equal and hence $\beta_d$ satisfies:
\begin{gather}\label{betad}
(d-1)\sin^2\frac{\beta_d}{2}+\sin^2\frac{(d-1)\beta_d}{2}=\frac{1}{8}.
\end{gather}
In order to apply Theorem \ref{mod} we need to check if $\frac{\beta_d}{2\pi}<\frac{1}{2d}$. By equation (\ref{betad}) $\beta_d$ is clearly close to zero and therefore we can assume that $\sin\frac{\beta_d}{2}$ approximately equals to $\frac{\beta_d}{2}$. Then it follows that $\frac{\beta_d}{2\pi}=\frac{1}{2\pi\sqrt{2d(d-1)}}$ which is clearly smaller than $\frac{1}{2d}$. Thus we can apply Theorem \ref{mod} to the lattice  $\mc{L}_{d-1,d}$ and the point $a=(a_1,\ldots,a_{d-1})$ with $\epsilon=\frac{\beta_d}{2\pi}< \frac{1}{2d}$. As a result we obtain point $p\in \mc{L}_{d-1,d}$ such that:
\begin{gather}\label{dirichlet_many}
|na_i-p_i|<\frac{\beta_d}{2\pi},
\end{gather}
where 
\[
n< \ceil*{\frac{1}{d}\left(\frac{2\pi}{\beta_d}\right)^{d-1}}.
\]
The result follows.
\end{proof}
For $d=3$ we obtain $\frac{\beta_3}{2}=\arctan\sqrt{\frac{6 - \sqrt{34}}{2 + \sqrt{34}}}$ and $N_{SU(3)}<154$. On the other hand numerical calculations yield $N_{SU(3)}= 49$. For orthogonal groups we have that numerical calculations yield $N_{SO(5)}=172$ and  $N_{SO(4)}=86$, where the bounds given by (\ref{nmax_so2kp1}) and (\ref{nmax_so2k}) are $N_{SO(5)}< 312$ and $N_{SO(4)}<151$  respectively. The difference between the bounds and values calculated numerically reflects the obvious fact, that the considered hypercubes are rather brutal approximations of the balls $B_{\alpha }$ (see figure \ref{fig4}). However, we stress that the choice of hypercubes we made is the most optimal from the perspective of Dirichlet's theorems. Let us also note that the upper bound for $N_G$ seems to be more accurate for $SO(4)$ than for $SU(3)$. We believe this stems from the fact that the 'square-ball` area ratio is smaller for $SU(3)$ than for $SO(4)$ (see figure \ref{fig4}). The way how these ratios should be incorporated into formulas for the upper bound on $N_G$ is left as an open problem. We suppose this should be done by introducing some additional factor that depends on the  square-ball ratio.
\begin{figure}[ht!]
\begin{center}
\includegraphics[scale=0.45]{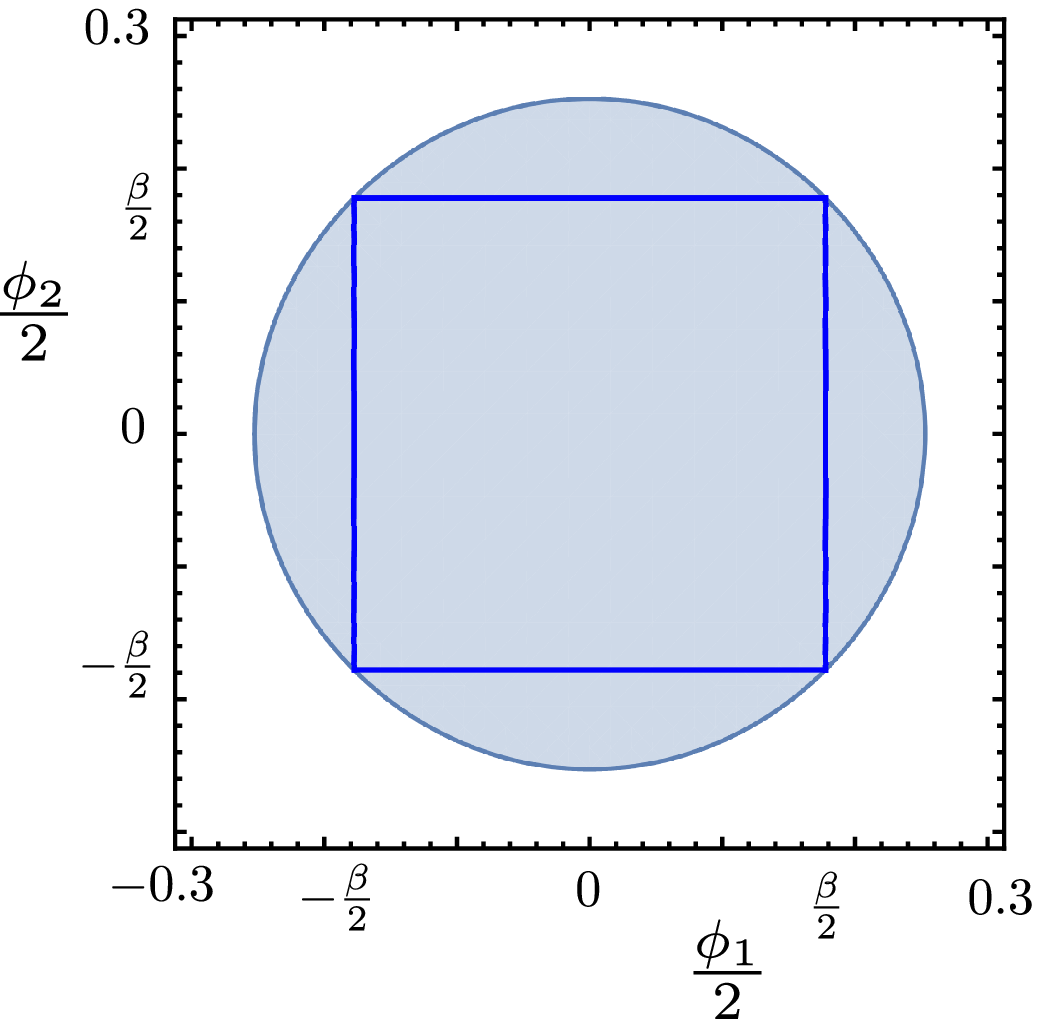}\quad{}\quad{}\quad{}\includegraphics[scale=0.45]{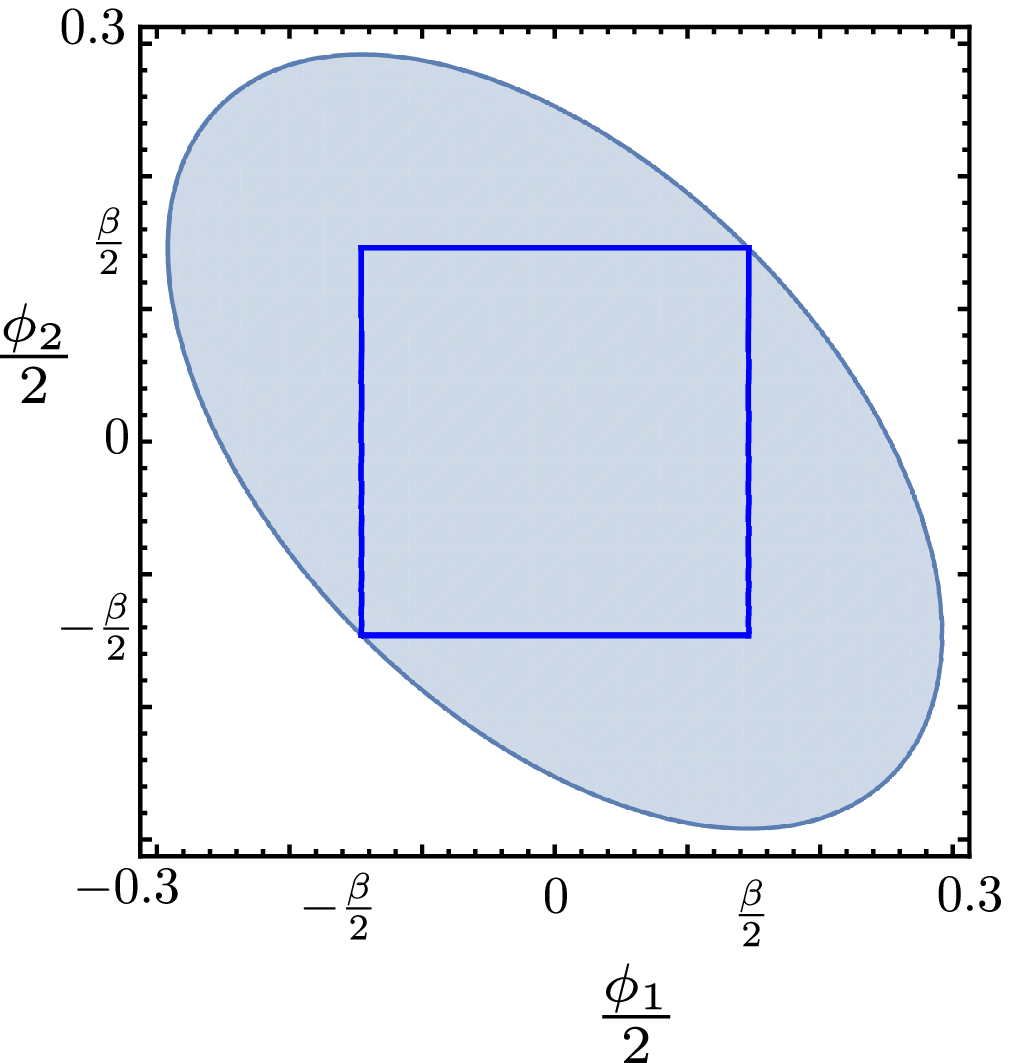}
\end{center}
\caption{\label{fig4}The smallest hypercubes contained in the balls $B_{1}$ for $SO(4)$ and $SU(3)$ respectively.}
\end{figure}

\section{Universality for $SU(2)$ and $SO(3)$}\label{kubity}
In the following we discuss universality of gates in case when $G=SU(2)$ or $G=SO(3)$. In particular we formulate explicit conditions (Fact \ref{solsu2}) for $\mathcal{C}(\mr{Ad}_\mc{S})=\{\lambda I\}$, where $\mc{S}$ is a finite subset of $G$.  For both $SU(2)$ and $SO(3)$ exceptional spectra are determined by one spectral angle and if at least one matrix from $\mc{S}$ has nonexceptional spectrum the algorithm from section \ref{alg} terminates in Step 2 with $l=1$. In section \ref{2exp} we show that for $\mc{S}$ consisting of two matrices that have exceptional spectra one can decide their universality in at most $l=4$ steps. Moreover, our algorithm always terminates for $l\leq 13$.

\subsection{$SU(2)$ and $SO(3)$ - review of useful properties}\label{sec:su2_su3_properties}
In the following we recall useful facts about groups $SO(3)$ and $SU(2)$. In particular we introduce their parameterizations and briefly discuss the covering homomorphism given by the adjoint representation. 

Commutation relations for the Lie algebras of the considered groups are as follows:
\begin{gather}
\mathfrak{su}(2):\,\,\,\left[X,Y\right]=2Z,\,\,\left[X,Z\right]=-2Y,\,\,\left[Y,Z\right]=2X.
\end{gather}where ${X},\;{Y},\;{Z}$ are defined as
\begin{gather*}
{X}=\left(\begin{array}{cc}
0&1\\-1&0\end{array} \right),\:{Y}=\left(\begin{array}{cc}
0&i\\i&0\end{array} \right),\:{Z}=\left(\begin{array}{cc}
i&0\\0&-i\end{array} \right),
\end{gather*}and
\begin{gather}
\mathfrak{so}(3):\,\,\,\left[X_{23},X_{13}\right]=-X_{12},\,\,\left[X_{23},X_{12}\right]=X_{13},\,\,\left[X_{13},X_{12}\right]=X_{23},
\label{eq:commutators}
\end{gather}
where $X_{ij}$ are defined as in (\ref{Xij}). The Lie algebras $\mathfrak{su}(2)$ and $\mathfrak{so}(3)$ are isomorphic through the adjoint representation $\mathrm{ad}:\mathfrak{su}(2)\rightarrow \mathfrak{so}(3)$. The isomorphism is established by ${X}\mapsto \mathrm{ad}_{{X}}=-2X_{23}$, ${Y}\mapsto \mathrm{ad}_{{Y}}=2X_{13}$, ${Z}\mapsto \mathrm{ad}_{{Z}}=-2X_{12}$.

Elements of groups $SU(2)$ and $SO(3)$ can be expressed using exponential map. By Cayley-Hamilton theorem we have:
\begin{gather}
SU(2):\,\,U(\phi,\vec{k})=e^{ \phi\cdot u(\vec{k})}=e^{\phi\left(k_xX+k_yY+k_zZ\right)}=\cos\phi I+\sin\phi u(\vec{k}),\label{su2}\end{gather}
\begin{equation}
\begin{array}{c}
SO(3):\,\,O(\phi,\vec{k})=e^{\phi\cdot o(\vec{k})}=e^{\phi\left(-k_xX_{23}+k_yX_{13}-k_zX_{12}\right)}=\\I+\sin\phi o(\vec{k})-2\sin^2\frac{\phi}{2}o(\vec{k})^2,\label{so3}
\end{array}
\end{equation}
where $\vec{k}=[k_x,k_y,k_z]\in\mathbb{R}^3$ is a rotation axis, $k_x^2+k_y^2+k_z^2=1$ and $\phi\in [0,2\phi)$. Groups $SU(2)$ and $SO(3)$ are related by the covering homomorphism $\mathrm{Ad}:SU(2)\rightarrow SO(3)$ given by $\mathrm{Ad}_{e^A}=e^{\mathrm{ad}_A}$, where $A\in \mathfrak{su}(2)$ and $\mr{Ad}:U(\phi,\vec{k})\mapsto O(2\phi,\vec{k})$. $\mathrm{Ad}$ is in this case double covering. Using (\ref{su2}) we can easily calculate the product $U(\gamma,\vec{k}_{12})=U(\phi_1,\vec{k}_1)U(\phi_2,\vec{k}_2)$, where:
\begin{gather}\label{product}
\cos\gamma = \cos\phi_1\cos\phi_2-\sin\phi_1\sin\phi_2 \vec{k}_1\cdot\vec{k}_2,\\
\vec{k}_{12}=\frac{1}{\sin\gamma}\left(\vec{k}_1\sin\phi_1\cos\phi_2+\vec{k}_2\sin\phi_2\cos\phi_1 + \vec{k}_1\times \vec{k}_2\sin\phi_1\sin\phi_2 \right).
\end{gather}
Making use of (\ref{product}) one checks that two $SU(2)$ matrices $U(\phi_1,\vec{k}_1)$, $U(\phi_2,\vec{k}_2)$ that do not belong to $\{I,-I \}$ commute iff the axes $\vec{k}_1$ and $\vec{k}_2$ are parallel, that is  $[u(\vec{k}_1),u(\vec{k}_2)]= 0$.  Similarly, they anticommute iff the axes $\vec{k}_1$ and $\vec{k}_2$ are orthogonal and  rotation angles are $\phi_1, \phi_2\in\{\frac{\pi}{2},\frac{3\pi}{2}\}$. As for matrices from $SO(3)$, recall that they cannot anticommute. In order to check when they commute we note, that commuting and anticommuting $SU(2)$ matrices satisfy the identity $U_1U_2U_1^{-1}U_2^{-1}=\pm I$. But $\mathrm{Ad}_{\pm I}=I$ and therefore $O(\phi_1,\vec{k}_1)$ commutes with $O(\phi_2,\vec{k}_2)$ iff either axes $\vec{k}_1$ and $\vec{k}_2$ are parallel or $\vec{k}_1\perp\vec{k}_2$ and $\phi_1, \phi_2\in\{\frac{\pi}{2},\frac{3\pi}{2}\}$.

Finally it is known that all automorphisms of $SU(2)$ are inner authomorphisms, thus they are in one to one correspondence  with elements of $SO(3)$. Using our notation $O\in SO(3)$ determines the automorphism $\Phi_O:SU(2)\rightarrow SU(2)$ given by
\begin{gather}\label{auto}
\Phi_O(U(\phi,\vec{k}))=U(\phi,O\vec{k}).
\end{gather}

\subsection{Exceptional spectra and spaces $\mc{C}(\mr{Ad}_{\mc{S}})$ for $SU(2)$ and $SO(3)$}\label{sec:exceptionalSpectraSu2So3}

For any matrix $U(\phi, \vec{k})\in SU(2)$ the spectrum is given by $\{e^{i\phi},e^{-i\phi}\}$, $\phi\in [0,2\pi)$. By Definition \ref{exangle} the spectrum of $U(\phi,\vec{k})$ is exceptional iff $e^{i\phi}$ is a root of $1$ or $-1$ of order $n=\{1,\ldots,N_{SU(2)}\}$. The corresponding $\phi\in[0,2\pi)$ will be called exceptional angle. Similarly for $O(\phi,\vec{k})\in SO(3)$ the spectrum is given by $\{e^{i\phi},e^{-i\phi},1\}$ and thus is  exceptional iff $e^{i\phi}$ is a root of unity of order $1\leq n\leq N_{SO(3)}$. The corresponding $\phi\in[0,2\pi)$ will be called an exceptional angle. We can easily compute the number of exceptional spectra for $SU(2)$ and $SO(3)$ using the Euler totient function $\varphi(n)$ by noting that the roots of $-1$ of order $n$ are the roots of unity of order $2n$. 

Let us denote the sets of exceptional angles for $SU(2)$ and $SO(3)$  by $\mathcal{L}_{SU(2)}$ and $\mathcal{L}_{SO(3)}$ respectively. We have:   
\begin{gather}
|\mathcal{L}_{SU(2)}|=\sum_{n=1}^{6}\varphi(n)+\sum_{n=4}^{6}\varphi(2n)=24,\\
|\mathcal{L}_{SO(3)}|=\sum_{n=1}^{12}\varphi(n) = 46.
\end{gather}
The elements of sets $\mc{L}_G$ are of the form $\mathcal{L}_{G}=\left\lbrace a\pi: a\in\mathcal{L}_G^\prime\right\rbrace$, where 
\begin{gather}
\mathcal{L}_{SU(2)}^\prime=\{0,\frac{1}{2},1,\frac{3}{2},\frac{1}{3},\frac{2}{3},\frac{4}{3},\frac{5}{3},\frac{1}{4},\frac{3}{4},
\frac{5}{4},\frac{7}{4},\frac{1}{5},\frac{2}{5},\frac{3}{5},
\frac{4}{5},\frac{6}{5},\frac{7}{5},\frac{8}{5},\frac{9}{5},\frac{1}{6},\frac{5}{6},\frac{7}{6},\frac{11}{6}\},\nonumber\\
\mathcal{L}_{SO(3)}^\prime = \mathcal{L}_{SU(2)}^\prime
\cup\{\frac{2}{7},\frac{4}{7},\frac{6}{7},\frac{8}{7},\frac{10}{7},\frac{12}{7},\frac{2}{9},\frac{4}{9},\frac{8}{9},\frac{10}{9},\frac{14}{9},\frac{16}{9},\frac{2}{11},\frac{4}{11},\frac{6}{11},\frac{8}{11},\nonumber\\
\frac{10}{11},\frac{12}{11},\frac{14}{11},\frac{16}{11},\frac{18}{11},\frac{20}{11}\},
\label{bad_angles}
\end{gather}

We next discuss the conditions when the space $\mc{C}(\mathrm{Ad}_{U(\phi_1,\vec{k}_1)},\mathrm{Ad}_{U(\phi_2,\vec{k}_2)})$ is different than $\mc{C}(\mathrm{ad}_{u(\vec{k}_1)},\mathrm{ad}_{u(\vec{k}_2)})$. First, we note that elements $u(\vec{k}_1)$, $u(\vec{k}_2)$ generate Lie algebra $\mk{su}(2)$ iff $[u(\vec{k}_1),u(\vec{k}_2)]\neq 0$. In this case by Lemma \ref{ad1}, the solution set $\mc{C}(\mathrm{ad}_{u(\vec{k}_1)},\mathrm{ad}_{u(\vec{k}_2)})=\{\lambda I\}$. Using Fact \ref{diff-so} we note that the space $\mc{C}(\mathrm{Ad}_{U(\phi_1,\vec{k}_1)},\mathrm{Ad}_{U(\phi_1,\vec{k}_2)})$ can be different than $\mc{C}(\mathrm{ad}_{u(\vec{k}_1)},\mathrm{ad}_{u(\vec{k}_2)})$ if  at least  one $\phi_i$ is equal to $\frac{k\pi}{2}$. In the following we give exact conditions when it happens.
\begin{fact}\label{solsu2}
Assume that $[u(\vec{k}_1),u(\vec{k}_2)]\neq 0$. The space $\mc{C}(\mathrm{Ad}_{U(\phi_1,\vec{k}_1)},\mathrm{Ad}_{U(\phi_2,\vec{k}_2)})$ is larger than $\{\lambda I:\lambda \in \mathbb{R}\}$  if and only if: (1) $\phi_1,\phi_2\in\{\frac{\pi}{2}, \frac{3\pi}{2}\}$, (2) one of $\phi_i\in \{\frac{\pi}{2},\frac{3\pi}{2}\}$ and $\vec{k}_1\perp\vec{k}_2$.
\end{fact}
\begin{proof} By Fact \ref{diff-so} $\mc{C}(\mathrm{Ad}_{U(\phi_1,\vec{k}_1)},\mathrm{Ad}_{U(\phi_2,\vec{k}_2)})$ can be larger than $\{\lambda I:\lambda \in \mathbb{R}\}$ if at least one of the spectral angles of $\mathrm{Ad}_{U(\phi_1,\vec{k}_1)}$, $\mathrm{Ad}_{U(\phi_2,\vec{k}_2)}$ is $k\pi$. Therefore we have to consider situation when either two angles $\phi_1$ and $\phi_2$ are equal to $\frac{k\pi}{2}$ or exactly one of $\phi_i$'s is $\frac{k\pi}{2}$, $k\in\{1,3\}$. For the case (1) generators are of the form $U\left(\frac{k\pi}{2},\vec{k}_1\right)$ and $U\left(\frac{k\pi}{2},\vec{k}_2\right)$, where $\vec{k}_1\cdot\vec{k}_2$ is arbitrary.  Note that $\mathrm{Ad}_{U\left(\frac{k\pi}{2},\vec{k}_1\right)}=O(k \pi,\vec{k}_1)$ and  $\mathrm{Ad}_{U\left(\frac{k\pi}{2},\vec{k}_2\right)}=O(k \pi,\vec{k}_2)$ are rotation matrices by angles $2\phi_1=2\phi_2=k\pi$. A rotation $O(\phi_3,\vec{k}_{3})$  by an arbitrary angle $\phi_3$ and about the axis $\vec{k}_3=\vec{k}_1\times\vec{k}_2$ commutes with the rotations $O(k \pi,\vec{k}_1)$ and $O( k\pi,\vec{k}_2)$ and is different than $\lambda I$.

Let us consider the case when exactly one of $\phi_i$'s is $\frac{k\pi}{2}$. We are given the generators  $U\left(\frac{k\pi}{2},\vec{k}_1\right)$ and $U\left(\phi_2,\vec{k}_2\right)$.  Note that the rotation  $O(\pi,\vec{k})$, where $\vec{k}\parallel\vec{k}_2$, commutes with both $\mathrm{Ad}_{U\left(\frac{k\pi}{2},\vec{k}_1\right)}=O(k\pi,\vec{k}_1)$ and $\mathrm{Ad}_{U\left(\phi_2,\vec{k}_2\right)}=O(2\phi_2,\vec{k}_2)$ provided $\vec{k}_1\perp\vec{k}_2$. Therefore in this case $\mc{C}(\mathrm{Ad}_{U(\frac{k\pi}{2},\vec{k}_1)},\mathrm{Ad}_{U(\phi_2,\vec{k}_2)})$ is larger than $\{\lambda I:\lambda \in \mathbb{R}\}$. We are left with showing that if $\vec{k}_1\not\perp\vec{k}_2$ and exactly one $\phi_i$'s is an odd multiple of $\pi$, the space $\mc{C}(\mathrm{Ad}_{U(\frac{k\pi}{2},\vec{k}_1)},\mathrm{Ad}_{U(\phi_2,\vec{k}_2)})$ is equal to $\{\lambda I:\lambda \in \mathbb{R}\}$.

By formula (\ref{su2}) if $\vec{k}_1\not\perp\vec{k}_2$, $\phi_1=\frac{k\pi}{2}$ and $\phi_2\neq \frac{k\pi}{2}$, then the only orthogonal matrix commuting with $\mathrm{Ad}_{U(\frac{k\pi}{2},\vec{k}_1)}=O(k\pi,\vec{k}_1)$ and $\mathrm{Ad}_{U(\phi_2,\vec{k}_2)}$ is the identity matrix. In the following we show that relaxing orthogonality to an arbitrary endomorphism gives only $\lambda I$. To see this, note that endomorphisms commuting with $\mathrm{Ad}_{U(\phi_2,\vec{k}_2)}$ are of the form 
\[
A=\alpha_2 O(\theta_2,\vec{k}_2)+\beta_2\kb{\vec{k}_2}{\vec{k}_2},
\]
where $\alpha_2,\beta_2\in \mathbb{R}$ and $\theta_2\in[0,2\pi)$. On the other hand matrices commuting with $\mathrm{Ad}_{U(\frac{k\pi}{2},\vec{k}_1)}$ are of the form
\[
B=E(\vec{k}_1^\perp)+\beta_1\kb{\vec{k}_1}{\vec{k}_1},
\]
where $E(\vec{k}_1^\perp)$ is an arbitrary matrix acting on the $2$-dimensional space perpendicular to $\vec{k}_1$ such that $E(\vec{k}_1^\perp)\vec{k}_1=0$ and $\beta_1\in \mathbb{R}$. Let $\{\vec{k}_1, \vec{k}_2,\vec{k}_{12}\}$, where $\vec{k}_{12}=\vec{k}_1\times \vec{k}_2$ be a basis of $\mathbb{R}^3$. Matrices $A$ and $B$ must agree on the basis vectors. This way we obtain the following equations:
\begin{gather}
\beta_1\vec{k}_1=\alpha_2O(\theta_2,\vec{k}_2)\vec{k}_1+\beta_2\bk{\vec{k}_1}{\vec{k}_2}\vec{k}_2,\\
(\alpha_2+\beta_2)\vec{k}_2=E(\vec{k}_1^\perp)\vec{k}_2+\beta_1\bk{\vec{k}_1}{\vec{k}_2}\vec{k}_1,\\
E(\vec{k}_1^\perp)\vec{k}_{12}=\alpha_2O(\theta_2,\vec{k}_2)\vec{k}_{12}.
\end{gather}
The left hand side of (\ref{bad_angles}) is a vector perpendicular to $\vec{k}_1$ and the right hand side of (\ref{bad_angles}) is a vector perpendicular to $\vec{k}_2$. The only vector satisfying both of these conditions is proportional to $\vec{k}_{12}$ and therefore $\theta_2=n\pi$. Hence $O(\theta_2,\vec{k}_2)=\pm I$. From equation (\ref{bad_angles}) we get 
\[\beta_1\vec{k}_1=\pm\alpha_2\vec{k}_1+\beta_2\bk{\vec{k}_1}{\vec{k}_2}\vec{k}_2,\]
which means $\beta_1=\pm\alpha_2$ and either $\beta_2=0$ or $\vec{k}_1\perp\vec{k}_2$. If $\beta_2=0$ then $A=\pm\alpha_2 I$ and hence the equality between $A$ and $B$ implies $$\mc{C}(\mathrm{Ad}_{U(\frac{k\pi}{2},\vec{k}_1)},\mathrm{Ad}_{U(\phi_2,\vec{k}_2)})=\{\lambda I:\lambda \in \mathbb{R}\}.$$ Therefore the only solution that yields a  bigger space $\mc{C}(\mathrm{Ad}_{U(\frac{k\pi}{2},\vec{k}_1)},\mathrm{Ad}_{U(\phi_2,\vec{k}_2)})$ corresponds to $\vec{k}_1\perp\vec{k}_2$. 
\end{proof} 
\subsection{Universal $SU(2)$ gates}\label{sec:su2}
In this section we consider the set $\mc{S}$ of  two noncommuting matrices $U(\phi_1,\vec{k}_1)$, $U(\phi_2,\vec{k}_2)$ and ask when they generate $SU(2)$. We treat separately three cases:
\begin{enumerate}
\item  When $\mc{C}(\mathrm{Ad}_{U(\phi_1,\vec{k}_1)},\mathrm{Ad}_{U(\phi_2,\vec{k}_2)})=\{\lambda I\}$ and at least one of $\phi_i$'s is nonexceptional - by Theorem \ref{main}, $\overline{<\mc{S}>}=SU(2)$, 
\item When $\mc{C}(\mathrm{Ad}_{U(\phi_1,\vec{k}_1)},\mathrm{Ad}_{U(\phi_2,\vec{k}_2)})=\{\lambda I\}$ and both angles are exceptional. This determines the maximal running time of the algorithm from section \ref{alg} to be $l=13$. 
\item When $\mc{C}(\mathrm{Ad}_{\mc{S}})\neq\{\lambda I\}$ we identify what is the structure of $<\mc{S}>$.
\end{enumerate}
%

We start from studying the last case. We already know that when $\vec{k}_1\perp\vec{k}_2$ and  $\phi_2=\frac{m\pi}{2}$,  where $m\in\{1,3\}$ the group generated by  $U(\phi_1,\vec{k}_1)$, $U(\phi_2,\vec{k}_2)$ is not $SU(2)$ as $\mc{C}(\mathrm{Ad}_{U(\phi_1,\vec{k}_1)},\mathrm{Ad}_{U(\phi_2,\vec{k}_2)})\neq\{\lambda I\}$. We will now show that in this case  this group is either finite or infinite {\it dicyclic group}.  To this end let $b:=U(\phi_1,\vec{k}_1)$ and $x:=U(\frac{\pi}{2},\vec{k}_2)$ and assume $b$ is of finite order. The group generated by $b$ and $x$ has the following presentation:
\begin{gather}
H=<b,x|\, x^4=I,\,b^n=I,\,xbx^{-1}=b^{-1}>.
\end{gather}
As $H$ contains $-I$ we have $(-b)^n=-I$ for $n$ odd. Let $a=-b$ then
\begin{gather}
H= <a,x|\, x^4=I,\,a^{2n}=I,\,xax^{-1}=a^{-1}>,
\end{gather}
which is the definition of the dicyclic group of order $4n$ (it is the central extension of the dihedral group of order $2n$ by $-I$). In case when $a$ is of the infinite order, after closure, we obtain a group consisting of two connected components. The first one is a one parameter group $\{U(t,\vec{k}_1):t\in\mathbb{R}\}$ generated by $U(\phi_1,\vec{k}_1)$ and the second one is its normaliser $\{U(\frac{\pi}{2},\vec{k}_2)U(t,\vec{k}_1):t\in\mathbb{R}\}$. The only other case when $\mc{C}(\mathrm{Ad}_{U(\phi_1,\vec{k}_1)},\mathrm{Ad}_{U(\phi_2,\vec{k}_2)})\neq\{\lambda I\}$ corresponds to the situation when both $\phi_1$ and $\phi_2$ are odd multiples of $\frac{\pi}{2}$. In this case the group generated by $U(\phi_1,\vec{k}_1)$, $U(\phi_2,\vec{k}_2)$ is the same as the group generated by $U(\gamma,\vec{k}_{12})=U(\phi_1,\vec{k}_1)U(\phi_2,\vec{k}_2)$ and $U(\phi_2,\vec{k}_2)$. One can easily calculate that $\cos\gamma=\vec{k}_1\cdot\vec{k}_2$ and $\vec{k}_{12}\perp\vec{k}_2$. Thus the group is once again the dicyclic group of the order $4n$ where $n$ is the order of $U(\gamma,\vec{k}_{12})$.
\begin{lem}\label{lem:permut}
Assume that $U(\phi_1,\vec{k}_1)$ and $U(\phi_2,\vec{k}_2)$ do not commute and  $\vec{k}_1\cdot\vec{k}_2= 0$ and $\phi_2\in\{\frac{\pi}{2},\frac{3\pi}{2}\}$. Then the group generated by $U(\phi_1,\vec{k}_1)$ and $U(\phi_2,\vec{k}_2)$  is either 1)the dicyclic group of order $4n$
\begin{gather}
n=\mathrm{max}(\mathrm{order}U(\phi_1,\vec{k}_1),\mathrm{order}U(\phi_1+\pi,\vec{k}_1),
\end{gather}
when $\mathrm{order}U(\phi_1,\vec{k}_1)<\infty$ or 2) the infinite dicyclic group if $\mathrm{order}U(\phi_1,\vec{k}_1)=\infty$. When both $\phi_i$'s belong to $\{\frac{\pi}{2},\frac{3\pi}{2}\}$ the group generated by $U(\phi_1,\vec{k}_1)$ and $U(\phi_2,\vec{k}_2)$ is also the dicyclic group of the order $4n$ where $n$ is the order of $U(\gamma,\vec{k}_{12})=U(\phi_1,\vec{k}_1)U(\phi_2,\vec{k}_2)$.
\end{lem}
In other words, the group generated by two noncommuting matrices from $SU(2)$ that do not satisfy the necessary condition for universality is either a finite or an infinite dicyclic group. 

\subsubsection{Two exceptional angles}\label{2exp}

Let $\phi_1\in\mathcal{L}_{SU(2)}\setminus\{0,\frac{\pi}{2},\pi,\frac{3\pi}{2}\}$\footnote{The case when both $\phi_i$'s are odd multiples of $\frac{\pi}{2}$ was treated in lemma \ref{lem:permut}.} and $\phi_2\in\mathcal{L}_{SU(2)}\setminus\{ 0,\pi\}$ and let $\mathcal{S}=\{U(\phi_1,\vec{k}_1),U(\phi_2,\vec{k}_2)\}$ be a two-element subset of  $SU(2)$. Using automorphism (\ref{auto}), for any $O\in SO(3)$ the group generated by $\mc{S}$ is isomorphic with the group generated by $U(\phi_1,O\vec{k}_1)$ and $U(\phi_2,O\vec{k}_2)$. This freedom allows us to choose $O\in SO(3)$ such that $\vec{k}_1^\prime= O\vec{k}_1=[0,0,1]$ and $\vec{k}_2^\prime= O\vec{k}_2=[\sin\alpha,0,\cos\alpha]$, for some $\alpha\in[0,2\pi)$. Thus in the following we will work with matrices $\mc{S}^\prime=\{U(\phi_1,\vec{k}^\prime_1),U(\phi_2,\vec{k}^\prime_2)\}$. Our aim is to determine how long does it take for the algorithm from section \ref{alg} to decide the universality of $\mc{S}^\prime$. If the algorithm does not terminate with $l=1$ this means that the product of matrices from $\mc{S}$ have exceptional spectral angles. Thus using formula \ref{product}
\begin{gather}
\cos\alpha=\vec{k}^\prime_{1}\cdot\vec{k}^\prime_2 = \frac{\cos\phi_1\cos\phi_2-\cos\gamma}{\sin\phi_1\sin\phi_2},
\end{gather} 
for some $\gamma\in \mc{L}_{SU(2)}$. In order to determine all such cases we need to exclude all triplets  $\phi_1,\phi_2, \gamma$  that lead to $|\cos\alpha|\geq1$. For all remaining cases we run our algorithm with matrices $\mathcal{S}$. The termination results are as follows:
\begin{enumerate}
\item The algorithm terminates in Step 2 for $l\leq 4$ and the resulting group is $SU(2)$.
\item The algorithm terminates in Step 3 with $5\leq l \leq 6$ and the resulting group has $24$ elements and is isomorphic to the binary therahedral group $<2,3,3>:=\lbrace a,b,c|a^2=b^3=c^3=abc\rbrace$.
\item The algorithm terminates in Step 3 with $7\leq l \leq 8$ and the resulting group has $48$ elements and is isomorphic to the binary octahedral group $<2,3,4>:=\lbrace a,b,c|a^2=b^3=c^4=abc\rbrace$.
\item  The algorithm terminates in Step 3 with $8\leq l \leq 13$ and the resulting group has $120$ elements and is isomorphic to the binary  icosahedral group $<2,3,5>:=\lbrace a,b,c|a^2=b^3=c^5=abc\rbrace$.
\end{enumerate}
To be more precise among all $10560$ exceptional triplets $\phi_1,\phi_2, \gamma$ there is $4816$ satisfying $|\cos\alpha|<1$. The number of triplets $\phi_1,\phi_2, \gamma$ that give termination of the algorithm for the length of the word equal to $l$ and the resulting groups are presented in Table \ref{tab:1}. 

\begin{table}[h]
\centering 
\begin{tabular}{||c|c|c|c||}
\hline\hline $l$&Step&Number of triplets $\phi_1,\phi_2, \gamma$&Generated group\\\hline\hline 
$-$&$1$&$80$&dicyclic group\\\hline 
$3$&$2 $&$3232$&$SU(2)$\\\hline 
$4$& $2$& $160 $&$SU(2)$\\\hline 
$5$& $3$ &$56 $&$<2,3,3>$\\\hline 
$6$& $3$ &$40 $&$<2,3,3>$\\\hline 
$7$ & $3$&$144 $&$<2,3,4>$\\\hline
$8$&$3$&$80$&$<2,3,4>$\\\hline
$8$&$3$&$240$&$<2,3,5>$\\\hline
$9$&$3$ &$352 $&$<2,3,5>$\\\hline 
$10$&$3$&$288$&$<2,3,5>$\\\hline 
$11$&$3$&$32$&$<2,3,5>$\\\hline 
$12$&$3$&$80$&$<2,3,5>$\\\hline 
$13$&$3$&$32$&$<2,3,5>$ \\\hline\hline
\end{tabular}
\caption{\label{tab:1} The number of exceptional triplets $\phi_1,\phi_2, \gamma$ terminating the universality algorithm for different $l$'s.}
\end{table}

As a direct consequence we get the following theorem:
\begin{thm}
\label{thm:two-mode-uni}Assume $\mc{S}=\{U(\phi_1,\vec{k}_1),U(\phi_2,\vec{k}_2)\}\subset SU(2)$. In order to verify universality of $\mc{S}$ it is enough to consider words of the length $l\leq 4$. Moreover, the algorithm terminates for $l\leq 13$. If it terminates in Step 1 the resulting group is either infinite or finite dicyclic group. If it terminates with $1\leq l\leq 4$ the resulting group is $SU(2)$. For $l\geq 5$ it is binary tetrahedral or binary octahedral or binary icosahedral group.
\end{thm}
\section{Universality of $2$-mode beamsplitters}\label{sec:2mode_in3modes}
In this section we address the universality problem of a single gate that belong to $SO(2)$ or $SU(2)$  and acts on a $d$-dimensional space, where $d>2$. More precisely, we consider the Hilbert space $\mathcal{H}=\mc{H}_1\oplus\ldots\oplus\mc{H}_d$, where $\mc{H}_k\simeq\mathbb{C}$, $d>2$. Next we take a matrix $B\in SU(2)$ or $B\in SO(2)$. This matrix will be referred to as a $2$-mode beamsplitter. We assume that we can permute {\it modes} and therefore we have access to matrices $B$ and $B^{\sigma}=\sigma^tB\sigma$, where $\sigma$ is the permutation matrix. Next, we define matrices $B_{ij}$ or $B^{\sigma}_{ij}$ to be the matrices that act on a $2$-dimensional subspace $\mc{H}_i\oplus\mc{H}_j\subset \mc{H}$ as $B$ or $B^{\sigma}$ respectively and on the other components of $\mc{H}$ as the identity. This way we obtain the set of $2{d \choose 2}=d(d-1)$ matrices $\mathcal{S}_d=\{B_{ij},B^{\sigma}_{ij}:i<j,\:i,j\in\{1,\ldots,d\}\}$ in $SU(d)$ or $SO(d)$ respectively. Let us denote by $\mc{X}_d=\{b_{ij}, b_{ij}^{\sigma}:i<j,\:i,j\in\{1,\ldots,d\}\}$ the set of corresponding Lie algebra elements $B_{ij}=e^{b_{ij}}$, $B_{ij}^{\sigma}=e^{b_{ij}^{\sigma}}$ (constructed as in Section \ref{log-construct}). Our goal is to find out when $\mathcal{S}_d$ is universal, i.e. when $\overline{<\mc{S}_d>}=SO(d)$ or $\overline{<\mc{S}_d>}=SU(d)$. In particular we focus on showing, for which $B$ the set $\mc{S}_3$ is universal. It is known that for such $B$ also any set $\mc{S}_d$ with $d>3$ will be universal (see \cite{Reck,sawicki1} for two alternative proofs).

\subsection{Spaces $\mc{C}(\mathrm{Ad}_{\mc{S}_3})$ and $\mc{C}(\mathrm{ad}_{\mc{X}_3})$}\label{sec:generating_Lie_algebra} 
In this section we characterise when $\mc{C}(\mathrm{Ad}_{\mc{S}_3})=\{\lambda I\}$ for both orthogonal and unitary beamsplitters. Our strategy is to first check when $\mc{C}(\mathrm{ad}_{\mc{X}_3})=\{\lambda I\}$. This can be done relatively easy. Then we use Facts \ref{diff-su} and \ref{diff-so}  to find $\mc{C}(\mathrm{Ad}_{\mc{S}_3})$. 
\subsubsection{The case of orthogonal group}
Let $B\in SO(2)$ be a rotation matrix by an angle $\phi\in (0,2\pi)$. Making use of the notation introduced in Section \ref{kubity} we have 
\begin{gather}
\mc{S}_3=\{B_{23}(\pm\phi),B_{13}(\pm\phi),B_{12}(\pm\phi)\},\\
\mc{X}_3=\{\pm \phi X_{23},\pm \phi X_{13},\pm \phi X_{12}\},
\end{gather}
where $B_{ij}(\pm\phi)$ correspond to the rotation matrices in three dimensions, i.e. $B_{12}(\phi)=O(\pm\phi,\vec{k}_z)$, $B_{13}(\phi)=O(\pm\phi,\vec{k}_y)$ and $B_{23}(\phi)=O(\pm\phi,\vec{k}_x)$, where $\vec{k}_x=[1,0,0]$, $\vec{k}_y=[0,1,0]$, $\vec{k}_z=[0,0,1]$ and matrices $X_{i,j}$ are defined by (\ref{Xij}). Note that matrices belonging to $\mc{X}$ form a basis of the Lie algebra $\mk{so}(3)$ iff $\phi\neq 0$. Therefore by Corollary \ref{ad1} we know that $\mc{C}(\mr{ad}_{\mc{X}_3})=\{\lambda I\}$. The adjoint matrices $\mr{Ad}_{O(\pm\phi,\vec{k}_i)}$ are again rotation matrices by angles $\pm\phi$ along axes $\vec{k}_{i}$. On the other hand, by Fact \ref{diff-so}  we know that $\mc{C}(\mathrm{Ad}_{\mc{S}_3})$ can be different than $\mc{C}(\mathrm{ad}_{\mc{X}_3})$ only if $\phi=\pm\pi$. Indeed in this case the adjoint matrices $\mr{Ad}_{O(\pm\phi,\vec{k}_i)}$ commute. Summing up we have
\begin{fact}\label{Cso3}
For a $2$-mode orthogonal beamsplitter. If $\phi\neq 0$ then $\mc{C}(\mr{ad}_{\mc{X}_3})=\{\lambda I\}$. On the other hand $\mc{C}(\mr{Ad}_{\mc{S}_3})=\{\lambda I\}$ iff $\phi\notin\{ 0,\pi\}$.
\end{fact}

\subsubsection{The case of unitary group}
Let $B\in SU(2)$. Making use of the notation introduced in Section \ref{kubity} we assume $B=U(\phi,\vec{k})$, $\phi\neq 0\,\mr{mod}\,\pi$, $\vec{k}=[k_x,k_y,k_z]$ and $k_x^2+k_y^2+k_z^2=1$. Therefore we have: 

\begin{gather}
\mc{X}_3=\{b_{ij}, b_{ij}^\sigma:\,1\leq i<j\leq3\}\}=\phi\cdot\{k_x{X}_{ij}+k_y{Y}_{ij}+k_z{Z}_{ij}, \nonumber \\-k_x{X}_{ij}+k_y{Y}_{ij}-k_z{Z}_{ij}:\,1\leq i<j\leq3\},\\
\mc{S}_3=\{B_{ij},B_{ij}^\sigma:\,1\leq i<j\leq3\}\}=\{ I_{ij}(\phi)+\sin\phi(k_x{X}_{ij}+k_y{Y}_{ij}+k_z{Z}_{ij}),\nonumber\\ I_{ij}(\phi)+\sin\phi(-k_x{X}_{ij}+k_y{Y}_{ij}-k_z{Z}_{ij}):\,1\leq i<j\leq3\},
\end{gather}
where $I_{ij}(\phi)=\cos\phi (E_{ii}+E_{jj})+E_{ll}$, $l\in\{1,2,3\}\setminus\{i,j\}$ and matrices $\{X_{ij},Y_{ij},Z_{ij}\}$ are defined as in (\ref{Xij}). We start from finding $\mc{C}(\mr{ad}_{\mc{X}_3})$. To this end note that $\left[b_{ij},b_{ij}^\sigma\right]=4k_y\left(k_x{Z}_{ij}-k_z{X}_{ij}\right)$. If $\left[b_{ij},b_{ij}^\sigma\right]\neq 0$ then $b_{ij}$ and $b_{ij}^\sigma$ generate $\mk{su}(2)_{ij}$. Thus we have access to all elements ${{X}_{ij}}$, ${{Y}_{ij}}$ and ${{Z}_{ij}}$ $1\leq i<j\leq3$. Hence $\mc{X}_3$ generates $\mk{su}(3)$ and  $\mc{C}(\mr{ad}_{\mc{X}_3})=\{\lambda I\}$. 
If in turn $[b_{ij},b_{ij}^\sigma]=0$  then we need to consider four cases: (1) $k_y\neq0$ and $k_x=0=k_z$, (2) $k_y=0$ and  $k_x\neq0$ and $k_z\neq 0$, (3) $k_y=0=k_z$ and $k_x\neq0$, (4) $k_y=0=k_x$ and $k_z\neq0$. 
\begin{enumerate}
\item  In this case $b_{ij}=k_y{Y}_{ij}=b_{ij}^\sigma$, therefore we have access to
all $\{{Y}_{ij}\}_{i<j}$, $i,j\in\{1,2,3\}$. But by the commutation relations $[{Y}_{ij},{Y}_{ik}]=-{X}_{jk}$,
$[{Y}_{ij},{Y}_{jk}]=-{X}_{ik}$, $[{Y}_{ij},{Y}_{kj}]=-{X}_{ik}$ and $\left[{X}_{ij},{Y}_{ij}\right]=2{Z}_{ij}$. Thus we can generate all basis elements of $\mk{su}(3)$ starting from ${Y}_{ij}$'s. This means $\mc{C}(\mr{ad}_{\mc{X}_3})=\{\lambda I\}$. 

\item In this case $b_{ij}=-b_{ij}^\sigma$.
Direct calculations show that elements: 
\begin{gather*}
\left[b_{12},\left[b_{12},b_{13}\right]\right],\,\left[b_{12},\left[b_{12},b_{23}\right]\right],\,\left[b_{13},\left[b_{13},b_{12}\right]\right],\\
\left[b_{13},\left[b_{13},b_{23}\right]\right],\,\left[b_{23},\left[b_{23},b_{12}\right]\right],\,\left[b_{12},\left[b_{12},\left[b_{13},b_{23}\right]\right]\right],\\
\left[b_{23},\left[b_{13},\left[b_{23},b_{12}\right]\right]\right],\,\left[b_{13},\left[b_{13},\left[b_{23},b_{12}\right]\right]\right],
\end{gather*}
form a basis of $\mathfrak{su}(3)$. Thus $\mc{C}(\mr{ad}_{\mc{X}_3})=\{\lambda I\}$. 
\item In this case the algebra generated by $\mc{X}_3$ is clearly $\mk{so}(3)$. Hence $\mc{C}(\mr{ad}_{\mc{X}_3})\neq\{\lambda I\}$. 

\item In this case the algebra generated by $\mc{X}_3$ is abelian. Hence $\mc{C}(\mr{ad}_{\mc{X}_3})\neq\{\lambda I\}$. 
\end{enumerate}
We have just shown:
\begin{fact}\label{fact:su2_in_su3_algebra}
For a $2$-mode unitary beamsplitter $B=I\cos\phi+\sin\phi(k_x{X}+k_y{Y}+k_z{Z})$, where $k_x^2+k_y^2+k_z^2=1$ we have $\mc{C}(\mr{ad}_{\mc{X}_3})=\{\lambda I\}$ unless  (a) $k_y=0=k_z$ and $k_x=1$, (b) $k_y=0=k_x$ and $k_z=1$. 
\end{fact}
Next we characterise $\mc{C}(\mr{Ad}_{\mc{S}_3})$. The adjoint matrices $\mr{Ad}_{B_{ij}}$ and $\mr{Ad}_{B_{ij}^\sigma}$ are elements of $SO(\mk{su}(3))\simeq SO(8)$. The rotation angles of both $\mr{Ad}_{B_{ij}}$ and $\mr{Ad}_{B_{ij}^\sigma}$ are $\pm \phi$, $2\phi$ and $0$. On the other hand, by Fact \ref{diff-su} we know that $\mc{C}(\mathrm{Ad}_{\mc{S}_3})$ can be different than $\mc{C}(\mathrm{ad}_{\mc{X}_3})$ only if the rotation angle is $\pm\pi$. This corresponds to situations when either $\phi=\pm\pi$ or $\phi=\pm\frac{\pi}{2}$. In the first case $B=-I$, thus obviously $\mc{C}(\mathrm{Ad}_{\mc{S}_3})\neq \{\lambda I\}$. The case $\phi=\pm \frac{\pi}{2}$ corresponds to $\frac{\pi}{2}\cdot\mc{S}_3=\mc{X}_3$. 

\begin{fact}\label{fact:su2_in_su3_group}
For a $2$-mode unitary beamsplitter $B=I\cos\phi+\sin\phi(k_x{X}+k_y{Y}+k_z{Z})$ we have $\mc{C}(\mr{Ad}_{\mc{S}_3})=\{\lambda I\}$ unless  (a) $k_y=0=k_z$ and $k_x=1$, (b) $k_y=0=k_x$ and $k_z=1$, (c) $\phi=\pm\frac{\pi}{2}$ and $k_z=0$.
\end{fact}
\begin{proof}
Recall that $\mc{C}(\mr{ad}_{\mc{X}_3})\subseteq\mc{C}(\mr{Ad}_{\mc{S}_3})$. Cases (a) and (b) correspond to situations when $\mc{C}(\mr{ad}_{\mc{X}_3})\neq\{\lambda I\}$. Case (c) follows from direct calculations for six $\mr{Ad}_g$ matrices with $\phi=\pm\frac{\pi}{2}$ and $g\in \mc{S}_3$. They were done with the help of a symbolic calculation software. We only verify that when $\phi=\pm\frac{\pi}{2}$ and $k_z=0$ indeed $\mc{C}(\mr{Ad}_{\mc{S}_3})\neq\{\lambda I\}$. Therefor we define $\mk{h}=\mr{Span}_{\mathbb{R}}\{{Z}_{12},{Z}_{23}\}$, $\mr{dim}_\mathbb{R}\mk{h}=2$ and show that for $\phi=\pm\frac{\pi}{2}$ and $k_z=0$ the space $\mathfrak{h}$ is an invariant subspace for matrices  $\mr{Ad}_{B_{ij}}$ and $\mr{Ad}_{B_{ij}^\sigma}$, i.e. of $\mc{S}_3$. To this end we calculate 
\begin{gather}
\mr{Ad}_{B_{12}}{Z}_{12}=-{Z}_{12},\,\,\mr{Ad}_{B_{13}}Z_{12}=-{Z}_{23},\,\,\mr{Ad}_{B_{23}}{Z}_{12}={Z}_{12}+{Z}_{23},\\
\mr{Ad}_{B_{12}}{Z}_{23}={Z}_{23}+{Z}_{12},\,\,\mr{Ad}_{B_{13}}{Z}_{23}=-{Z}_{12},\,\,\mr{Ad}_{B_{23}}{Z}_{23}=-{Z}_{23}.
\end{gather}
and $\mr{Ad}_{B^\sigma_{ij}}{Z}_{kl}=\mr{Ad}_{B_{ij}}{Z}_{kl}$. Therefore the projection operator $P:\mk{su}(3)\rightarrow \mk{h}$ commutes with matrices from $\mc{S}_3$ and thus it belongs to $\mc{C}(\mr{Ad}_{\mc{S}_3})$.
\end{proof}

It is interesting to look at the structure of the group $\overline{<\mc{S}_3>}$ when $k_z=0$ and $\phi=\frac{\pi}{2}$.  Matrices are of the form $B_{ij}=e^{i\psi}E_{ij}-e^{-i\psi}E_{ji}+E_{kk}$ and $B_{ij}^\sigma=-e^{-i\psi}E_{ij}+e^{i\psi}E_{ji}+E_{kk}$, where $1\leq i<j\leq 3$, $k\neq i,j $ and $\psi\in[0,2\pi)$. If $\psi$ is a rational multiple of $\pi$, then it is easy to see that  $<\mc{S}_3>$ is a finite group and when $\psi$ is an irrational multiple of $\pi$ the group $\overline{<\mc{S}_3>}$ is infinite and disconnected. In fact these are groups isomorphic to $\Delta(6n^2)$ and $\Delta(6\infty^2)$ given in \cite{fulton}.
\subsection{When $\mc{S}_{3}$ is universal?}
Having characterised when $\mc{C}(\mr{Ad}_{\mc{S}_3})=\{\lambda I\}$ we check  in this section when the group $<\mc{S}_3>$ is infinite and this way we get the full classification of universal $2$-mode beamsplitters.
\subsubsection{The case of the orthogonal group}
\label{sec:so2_into_so3_matlab}
Combining Theorem \ref{main} with Fact \ref{Cso3} for $\phi\notin \mathcal{L}_{SO(3)}$ we obtain that the group generated by  $\mc{S}_3$ is exactly $SO(3)$. When $\phi\in\mathcal{L}_{SO(3)}$ we consider the matrix:  $O(\gamma,\vec{k}_{xz})=O(\phi,\vec{k}_x)O(\phi,\vec{k}_z)$. The trace  yields the following equation that relates $\gamma$ and $\phi$:
\begin{equation}
\cos\gamma=\frac{\cos^2\phi+2\cos\phi-1}{2}.
\label{def:gamma_so3}
\end{equation}
If $\phi = \frac{(2k+1)\pi}{2}$, where $k\in \mathbb{Z}$, then matrices $O(\phi,\vec{k}_x)$, $O(\phi,\vec{k}_y)$ and $O(\phi,\vec{k}_z)$ are permutation matrices and they form $3$-dimensional representation of $S_3$. For all remaining $\phi \in\mathcal{L}_{SO(3)}$ we calculate $\cos\gamma$ using (\ref{def:gamma_so3}) and compare it with the values of $\cos\alpha$ for all $\alpha\in\mathcal{L}_{SO(3)}$. We find out they never agree. Therefore $\gamma\notin\mathcal{L}_{SO(3)}$ and we can apply Theorem  \ref{main} and Fact \ref{Cso3} to $U(\gamma,\vec{k}_{xz})$. Summing up:
\begin{thm}\label{beam-real}
Any $2$-mode orthogonal beamsplitter with $\phi\notin\{\frac{\pi}{2},\frac{3\pi}{2}\}$ is universal on $3$ and hence $n>3$ modes.
\end{thm}
\subsubsection{The case of the unitary group}
Recall that by Fact \ref{fact:su2_in_su3_group}  the space $\mc{C}(\mr{Ad}_{\mc{S}_3})=\{\lambda I\}$ if and only if all the entries of a matrix $B\in SU(2)$ are nonzero and at least one of them belongs to $\mathbb{C}$. So we are left with checking if under these assumptions $<\mc{S}_3>$ is infinite. Let $\{e^{i\phi},e^{-i\phi}\}$ be the spectrum of $B$. Matrices $B_{ij}$ and $B_{ij}^\sigma$ have the same spectra $\{e^{i\phi},e^{-i\phi},1\}$. Looking at the definitions of the open balls $B_{\alpha}$, $\alpha^3=1$ we see that a matrix from $SU(3)$ with one spectral element equal to one can be introduced (by taking powers) only to the ball with $\alpha=1$. Moreover, the maximal $n$ that is needed is exactly the same as for $SO(3)$ and the exceptional angles belong to the set $\mc{L}_{SO(3)}$. Therefore, by Theorem \ref{main}, $\phi\notin\mc{L}_{SO(3)}$ implies that the group generated by, for example, $B_{12}$ and $B_{23}$ is infinite. In the following we show that $<\mc{S}_3>$ is infinite also for $\phi\in\mc{L}_{SO(3)}$ (providing $\phi$ is such that $\mc{C}(\mr{Ad}_{\mc{S}_3})=\{\lambda I\}$).

Let us consider $<\mc{R}>=<B_{12}(\phi),B_{23}(\phi)>$ with $\phi\in\mc{L}_{SO(3)}$. Our goal is to show that  $\mc{R}\subset\mc{S}_3$ generates an infinite group. To this end we use the following procedure:
\begin{enumerate}
\item We calculate trace of the product $B_{12}(\phi)B_{23}(\phi)$ and note that it belongs to $\mathbb{R}$. Therefore spectrum of $B_{12}(\phi)B_{23}(\phi)$ is of the form $\{e^{i\gamma},e^{-i\gamma},1\}$, where the relation between $\phi$ and $\gamma$ is given by
\begin{gather}\label{def:trace}
\mr{tr}B_{12}(\phi)B_{23}(\phi) = 2\cos\phi+\cos^2\phi+k_z^2\sin^2\phi=2\cos\gamma+1.
\end{gather} 
\item Using (\ref{def:trace}), for each $\gamma\in\mathcal{L}_{SO(3)}$ we compute
\begin{gather}\label{def:z2}
k_z^2=\frac{2\cos\gamma+1-2\cos\phi-\cos^2\phi}{\sin^2\phi},
\end{gather}
and check whether $0<k_z^2<1$. The pairs $(\phi,\gamma)$ that fails this test are excluded form the further considerations. We note that $k_z^2=1$ corresponds to diagonal matrices $B_{12}(\phi), B_{23}(\phi)$ and $k_z^2=0$ corresponds the situation when $\mc{C}(\mr{Ad}_{\mc{S}_3})\neq \{\lambda I\}$. 
\item For the pairs $(\phi,\gamma)$ that give $0<k_z^2<1$ we consider the matrix $U(\gamma^\prime)=B_{12}(2\phi)B_{23}(2\phi)$. Its trace is again real and we get 
\begin{gather}\label{def:trace3}
\mr{tr}B_{12}(2\phi)B_{23}(2\phi) = \frac{1}{2} (2+4\cos(2\phi)+(1-k_z^2)(\cos(4\phi)-1) )=2\cos\gamma^\prime+1,
\end{gather}
where $k_z^2$ is determined by $\phi$ and $\gamma$. Direct computations show that $\gamma^\prime\notin\mc{L}_{SO(3)}$ if $\phi\notin\left\lbrace \pm\frac{\pi}{2},\pm\frac{2\pi}{3} \right\rbrace $. We treat both of these cases separately.  
\item For $\phi=\pm\frac{2\pi}{3}$ and the fixed $k_z^2$ we consider yet another product of matrices $U(\gamma^{''})=B_{23}^2(\phi)B^2_{12}(\phi)B_{23}(\phi)B_{12}(\phi)$ with a real trace:
\begin{gather}\label{trace4}
\mr{tr}B_{23}^2(\phi)B^2_{12}(\phi)B_{23}(\phi)B_{12}(\phi) =\frac{1}{8}( \cos \phi+3 \cos (2 \phi)+4 \cos (3 \phi)+6 \cos (4 \phi)\nonumber\\\nonumber+4 \cos (5\phi)+\cos (6 \phi)-2)+32 k_z^4 \sin ^4\phi \cos ^2\phi+8 k_z^2 \sin ^2\phi (-2 \cos \phi+ \nonumber \\+4 \cos (2 \phi)+2 \cos (3 \phi)+\cos (4 \phi)+4))=2\cos\gamma^{''}.\end{gather}
Direct computations show that $\gamma^{''}\notin\mc{L}_{SO(3)}$, thus we are done for $\phi\in\mc{L}_{SO(3)}\backslash\{\frac{\pi}{2},-\frac{\pi}{2} \} $. The same composition for $U_{23}\left(\frac{\pi}{2} \right),U_{12}\left(\frac{\pi}{2} \right)$ may give a matrix of the spectral angle $\gamma=\pm\frac{2\pi}{3}$.
\end{enumerate}
For $\phi=\pm\frac{\pi}{2}$ an additional treatment is needed. It consists of three steps:
\begin{enumerate}
\item Assume  $B_{ij}\left(\frac{\pi}{2} \right)$ does not commute with its permutations $B_{ij}^\sigma(\frac{\pi}{2})$ for $1\leq i<j\leq 3$. In this case we can use  $B_{ij}(\gamma)=B_{ij}\left(\frac{\pi}{2} \right)B_{ij}^\sigma\left(\frac{\pi}{2} \right)$, $1\leq i<j\leq 3$ as the new set of generators. Note that the angle $\gamma$ depends on the trace of $B_{ij}\left(\frac{\pi}{2} \right)B_{ij}^\sigma\left(\frac{\pi}{2} \right)$ as $\cos\gamma=1 - 2 k_y^2$. Thus $\gamma\neq\pm\frac{\pi}{2}$ if $k_y^2\neq\frac{1}{2}$ and then we can apply the previous procedure to show that $<B_{12}(\gamma),B_{23}(\gamma)>$ is infinite. 
\item For $\phi=\pm\frac{\pi}{2}$ and $k_y^2=\frac{1}{2},\;k_x^2+k_z^2=\frac{1}{2}$ we consider yet another product 
\begin{gather*}
\mr{tr}B^2_{12}\left(\frac{\pi}{2}\right)B_{13}\left(\frac{\pi}{2}\right)B_{23}\left(\frac{\pi}{2}\right)B^2_{13}\left(\frac{\pi}{2}\right)=k_z^2=2\cos\gamma^{'''}
\end{gather*}We find out that the only $\gamma\in\mc{L}_{SO(3)}$ satisfying $2\cos\gamma=k_z^2-1$ for $0\leq k_z^2\leq \frac{1}{2}$ are $\gamma=\pm\frac{2\pi}{3}$. But then $k_z^2=0$. Thus by Fact  \ref{fact:su2_in_su3_group} the space $\mc{C}(\mr{Ad}_{\mc{S}_3})$ is larger than $\{\lambda I\}$.
\item Finally we assume that matrices $B_{ij}\left(\frac{\pi}{2} \right)$ commute with their permutations. Recall that it happens if either $k_y=\pm1$ and $k_x=k_z=0$ or $k_y=0$ and $k_x,k_z\neq 0$. The group generated for $k_y=\pm1$ is of course finite. Therefore we need to consider only the case when $k_y=0$ and $k_x,k_z\neq 0$. But in this case step 2 of the previous procedure is never satisfied (from equation (\ref{def:z2}) one can only obtain $k_z^2=0$ for $\gamma=\pm\frac{2\pi}{3}$).
\end{enumerate} 
Summing up:
\begin{thm}
Any $2$-mode unitary gate, such that all its entries are nonzero and at least one of them is a complex number is universal on $3$ and hence $n>3$ modes.
\end{thm}

\section*{Acknowledgments}

We would like to thank Tomasz Maci\k{a}\.{z}ek for fruitful discussions. AS would like to thank Bartosz Naskrecki for stimulating discussions concerning Dirichlet theorem, Adam Bouland and Laura Man\v{c}inska for two long meetings concerning universal Hamiltonians, Etienne Le Masson for his suggestion on adding Fact \ref{liesub}. KK would like to thank Daniel Burgath for his interest and comments. We would like to also thank the anonymous referees for
suggestions that led to improvements of the paper. This work was supported by National Science Centre, Poland under the grant SONATA BIS: 2015/18/E/ST1/00200. AS also acknowledges the support from the Marie Curie International Outgoing Fellowship.

\end{document}